%% file: bare_jrnl.tex
\documentclass[journal]{IEEEtran}
\ifCLASSINFOpdf
\else
\fi
\usepackage[utf8]{inputenc} % allow utf-8 input
\usepackage[T1]{fontenc}    % use 8-bit T1 fonts
\usepackage{hyperref}       % hyperlinks
\usepackage{url}            % simple URL typesetting
\usepackage{booktabs}       % professional-quality tables
\usepackage{amsfonts}       % blackboard math symbols
\usepackage{nicefrac}       % compact symbols for 1/2, etc.
\usepackage{microtype}      % microtypography
\usepackage{xcolor}         % colors

\usepackage{graphicx} % Required for inserting images
\usepackage{amssymb}
\usepackage{enumerate}
\usepackage{cite}
\usepackage{algorithm}
\usepackage{algorithmic}
\usepackage{amsmath}
\usepackage{enumitem}
\usepackage{amsthm}
\usepackage{bbm}
\usepackage{dsfont}
\usepackage[numbers,sort&compress]{natbib}
\usepackage{comment}
\usepackage{enumitem} 
\usepackage[caption=false,font=footnotesize]{subfig}
\usepackage{capt-of}

\newtheorem{remark}{Remark}
\newtheorem{definition}{Definition}
\newtheorem{theorem}{Theorem}
\newtheorem{lemma}{Lemma}

\newtheorem{assumption}{Assumption}
\newtheorem{problem}{Problem}

\begin{document}
\title{
% Orchestrating Teams in Game via Supervisor Network: Convergence and Robustness
Distributed Team Orchestration via Supervisor Networks: Convergence, Optimality, and Resilience 

}
\author{Juntian~Zhu, 
        Guanpu~Chen,
        Tongtian~Zhu,
        Miguel de Carvalho,
        Zhouwang Yang, and
        Fengxiang~He
\thanks{J. Zhu and Z. Yang are with the School of Artificial Intelligence and Data Science, University of Science and Technology of China, Hefei 230026, China. Email: zjt1229@mail.ustc.edu.cn, yangzw@ustc.edu.cn.}
\thanks{G. Chen is with the School of Automation, Southeast University, Nanjing 210096, China. Email: guanpu\_chen@seu.edu.cn.}
\thanks{T. Zhu is with the College of Computer Science and Technology, Zhejiang University, Hangzhou 310027, China. Email: raiden@zju.edu.cn.}
\thanks{M. de Carvalho is with the School of Mathematics, University of Edinburgh, Edinburgh EH9 3FD, Scotland. Email: Miguel.deCarvalho@ed.ac.uk.}
\thanks{F. He is with the School of Informatics, University of Edinburgh, Edinburgh EH8 9AB, Scotland. Email: fhe@ed.ac.uk.}
% \thanks{This work was partially conducted when J. Zhu was visiting the University of Edinburgh.}
% \thanks{Corresponding author: Guanpu Chen.}
}

\maketitle

\begin{abstract}
In this paper, we study zero-sum potential team games with a supervisor network, where agents rely on supervisor-provided belief information rather than accurate common beliefs.
The main challenge is that such belief information can be inaccurate because of supervisors' belief-estimation errors and the misreporting of joint actions by Byzantine teams.
We propose the distributed team-orchestrating algorithm (DTOA), which combines team fictitious play with supervisor-based distributed belief learning.
We prove the convergence of supervisors' belief estimates and establish that the induced learning dynamics converge to a near team-Nash equilibrium (TNE) in terms of the team-Nash gap (TNG). 
In the Byzantine setting, we consider a misreporting attack model and develop a Byzantine-resilient DTOA. 
We further provide probabilistic guarantees for Byzantine-team identification and establish an asymptotic bound on the honest TNG.
Numerical experiments illustrate the theoretical findings, compare DTOA with baseline learning methods, and evaluate its performance in a Markov decision process setting. 
\end{abstract}


% \fi
\IEEEpeerreviewmaketitle

\input{sections_IEEE/1_Introduction}
\input{sections_IEEE/2_Related_work}

\input{sections_IEEE/3_Preliminaries}
\input{sections_IEEE/4_Game_formulation}

\input{sections_IEEE/5_Convergence}

\input{sections_IEEE/6_Robustness}
\input{sections_IEEE/7_Experiments}

\input{sections_IEEE/8_Conclusion}

\bibliographystyle{IEEEtran}
\bibliography{reference}

%%%%%%%%%%%%%%%%%%%%%%%%%%%%%%%%%%%%%%%%%%%%%%%%%%%%%%%%%%%%

\appendix

\input{appendices_IEEE/proof_of_belief_convergence}
\input{appendices_IEEE/Proof_of_gap}
\input{appendices_IEEE/proof_of_byzantine}

\ifCLASSOPTIONcaptionsoff
  \newpage
\fi



%\bibliographystyle{IEEEtran}


\begin{IEEEbiography}[{\includegraphics[width=1in,height=1.25in,clip,keepaspectratio]{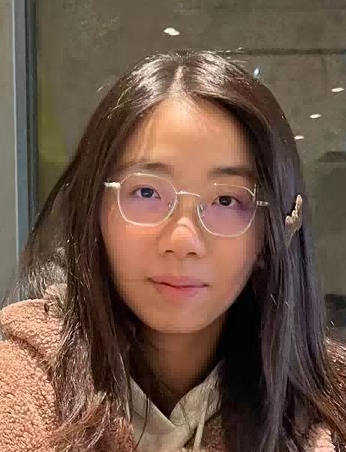}}]{Juntian Zhu}
is a Ph.D. candidate with the School of Artificial Intelligence and Data Science, University of Science and Technology of China.
She received the B.Sc. degree in mathematics from Sichuan University, Chengdu, China, in 2022.
Her research interests are in (1) learning theory in uncertain AI environments; (2) multi-agent learning, decision-making and game-theoretic analysis under limited information; and (3) uncertainty-aware reasoning in large language models.
\end{IEEEbiography}

\begin{IEEEbiography}[{\includegraphics[width=0.95in,height=1.2in,clip,keepaspectratio]{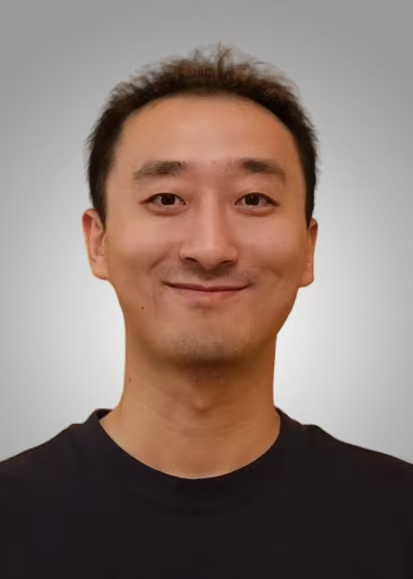}}]{Guanpu Chen} received his B.Sc. degree from University of Science and Technology of China, Hefei, China, in 2017, and Ph.D. degree from Academy of Mathematics and Systems Science, Chinese Academy of Sciences (CAS), Beijing, China, in 2022. He is currently a Professor with the School of Automation, Southeast University, Nanjing, China. He used to be a postdoctoral researcher with the School of Electrical Engineering and Computer Science, KTH Royal Institute of Technology, Stockholm, Sweden. His research interests include multi-agent systems, network games, as well as robustness, resilience, and security in cyber-physical systems. Prof. Chen was the recipient of the President Award of CAS, the Best Paper Award at IEEE International Conference on Control and Automation (ICCA), Guan Zhao-Zhi Award at Chinese Control Conference (CCC), and the Best Student Paper Honorable Mention at IEEE CSS TCSP. 
\end{IEEEbiography}

\begin{IEEEbiography}[{\includegraphics[width=1in,height=1.25in,clip,keepaspectratio]{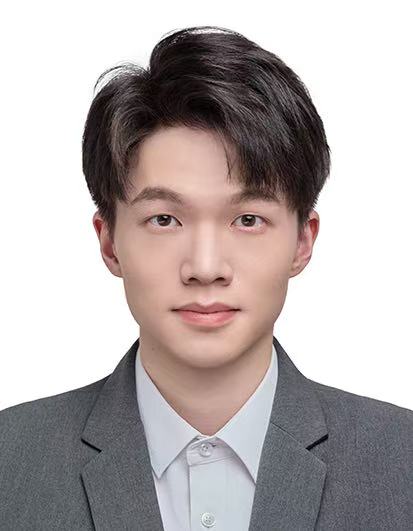}}]{Tongtian Zhu}
is a Ph.D. candidate with the College of Computer Science and Technology, Zhejiang University. His research interests are in (1) understanding modern deep learning systems from an optimization-theoretic perspective; and (2) the theoretical foundations and algorithms for scalable and autonomous decentralized, distributed, and multi-agent learning. He has made pioneering contributions to the study of implicit sharpness bias and data influence in decentralized learning, with related findings published at top-tier machine learning conferences, including ICML, NeurIPS, and ICLR, as oral and spotlight presentations.
\end{IEEEbiography}

\begin{IEEEbiography}[{\includegraphics[width=1in,height=1.25in,clip,keepaspectratio]{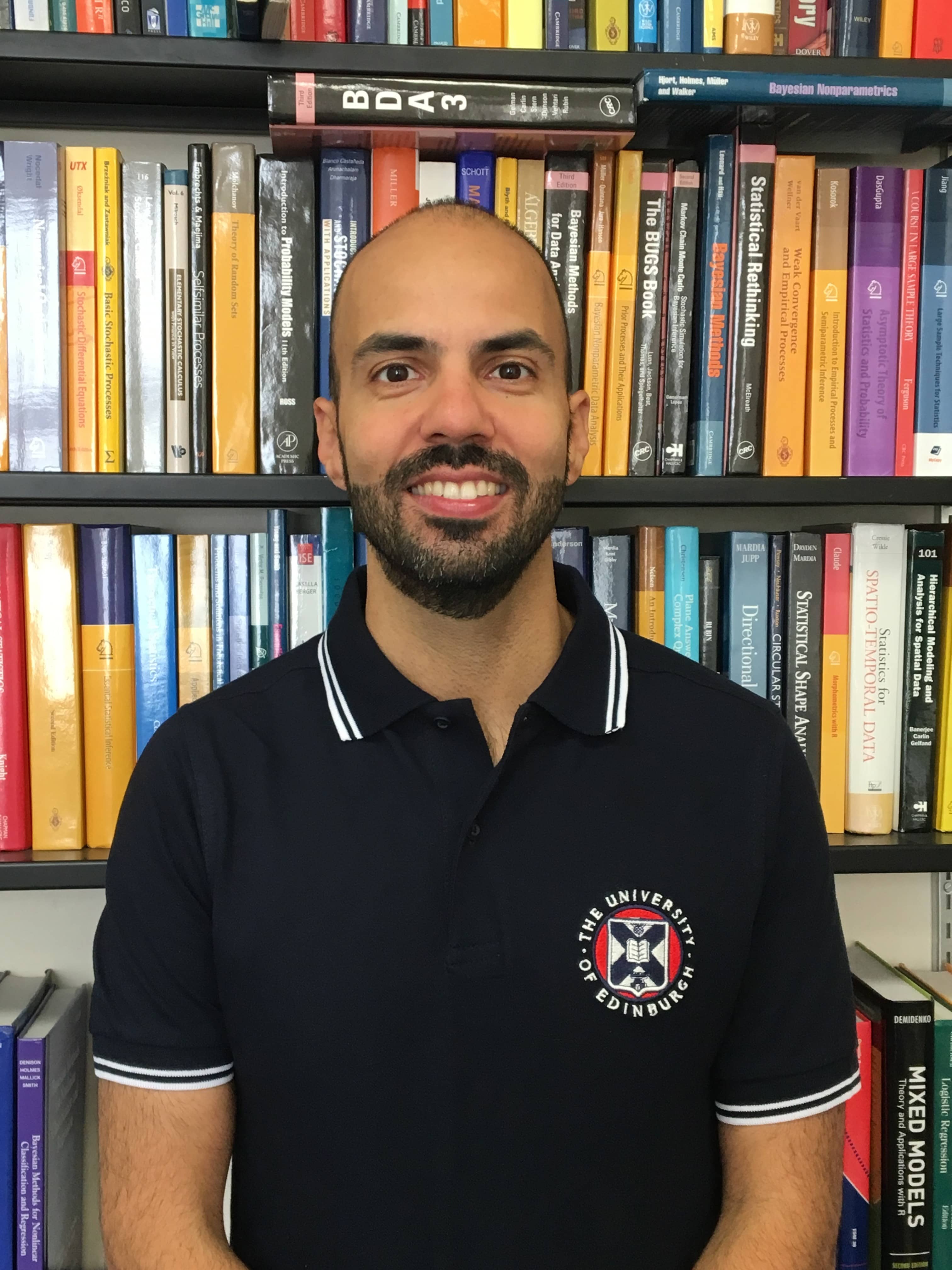}}]{Miguel de Carvalho}
is the Professor and Chair of Statistical Data Science, a Fellow of Generative AI Laboratory, and the Co-Director of the Edinburgh Centre for Financial Innovations at the University of Edinburgh. He is also an Honorary Professor at Universidade de Aveiro. He received a BSc in Mathematics from the NOVA School of Science and Technology in 2004, an MSc in Economics from the NOVA School of Business and Economics in 2009, and a PhD in Mathematics with emphasis on Statistics from the NOVA School of Science and Technology in 2009. His research is in Extreme Value Theory, Interfaces between Statistics and AI, and Bayesian Analysis. He is the Editor-in-Chief of the Springer book series on Courses in Advanced Statistics and Data Science, and an Associate Editor of American Statistician, the Annals of Applied Statistics, Bayesian Analysis, Computational Statistics \& Data Analysis, Extremes, the Journal of the American Statistical Association, and Statistics and Public Policy. 
\end{IEEEbiography}

\begin{IEEEbiography}[{\includegraphics[width=1in,height=1.25in,clip,keepaspectratio]{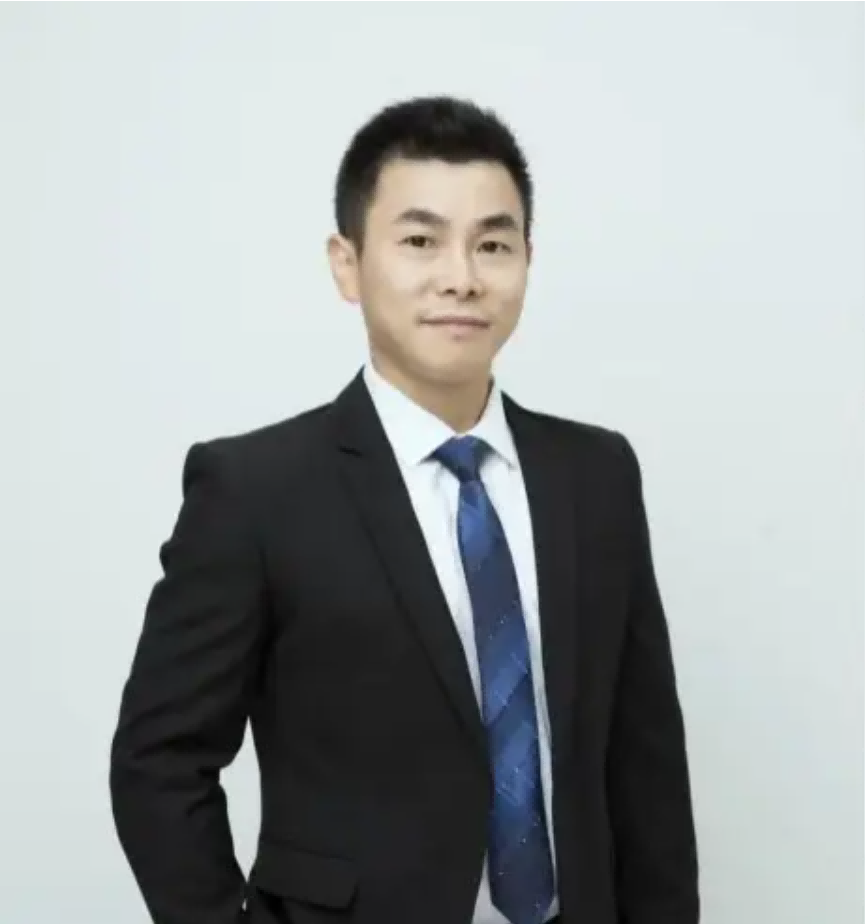}}]{Zhouwang Yang}
is a Professor in the School of Mathematical Sciences at University of Science and Technology of China (USTC). He received his Bachelor degree, Master degree and PhD degree in Mathematics from USTC in 1997, 2000 and 2005, respectively. He worked at Seoul National University as a postdoctoral researcher from December 2006 to November 2007. He was a visiting scholar in School of Industrial and Systems Engineering (ISyE) at Georgia Institute of Technology during August 2010 - August 2011. He has been working on geometric modeling and processing; optimization methods and their applications in sparse recovery. He also has a particular interest in the integration of optimization, machine learning and statistics for solving big data problems.
\end{IEEEbiography}


%\newpage

\begin{IEEEbiography}[{\includegraphics[width=1in,height=1.25in,clip,keepaspectratio]{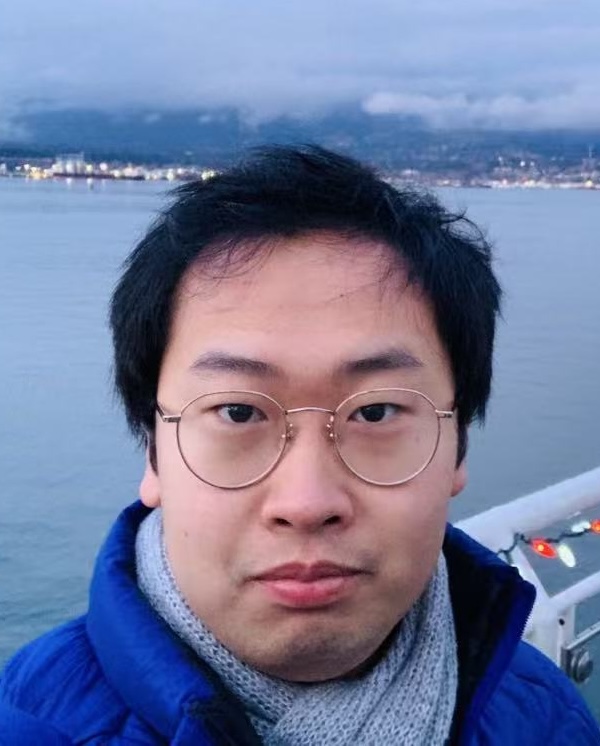}}]{Fengxiang He}
is a Lecturer at the University of Edinburgh, and a Fellow of its Generative AI Laboratory. He received a BSc in statistics from the University of Science and Technology of China, an MPhil and a PhD in computer science from the University of Sydney in 2017, 2019, and 2021, respectively. 
His research interests are in (1) understanding AI from both learning-theoretical and game-theoretical views; (2) new models and algorithms leveraging symmetries in data, task, environment; (3) collaboration and interactions between AI agents in decentralized and multi-agent settings; and (4) applications in economics and finance.
He is an Area Chair of ICML, NeurIPS, ICLR, UAI, and AISTATS, an Editorial Board member of Machine Learning, and an Associate Editor of Pattern Recognition. 
\end{IEEEbiography}

\vfill


%\vfill


\end{document}

%% file: sections_IEEE/1_Introduction.tex
\section{Introduction}

% \vspace{-3pt}
With the rapid development of artificial intelligence (AI), decision-making problems in multi-agent systems have received increasing attention.
A large body of work has investigated cooperative and competitive interactions in multi-agent systems, which can be broadly classified into three categories: purely cooperative single-team problems, single-team problems involving external adversaries, and multi-team games.
Single-team problems, with or without external adversaries, consider settings in which all agents cooperate to achieve a common objective \citep{1102542,9844798,11226888,11239499}.
However, many complex decision-making environments involve multiple teams whose members cooperate internally while competing strategically with other teams, and such settings are naturally modeled by multi-team games \citep{kim2022games,donmez2024team,xiao2025solving}.

% Purely cooperative single-team problems, in which all agents share a common objective and face no external adversary, have been extensively studied, from early continuous-time and stochastic formulations 
% \citep{1102542}
% % \citep{1102542,1104636,121618} 
% to subsequent developments in sequential team problems \citep{9844798} and average tracking
% \citep{11226888}.
% % \citep{11226888,10491342,9802705}
% % , and team learning and task allocation \cite{9729464,8115275}.
% Studies on single-team problems with external adversaries have mainly focused on security settings 
% %\cite{11239499,9049118,10772632,10731572,9204421} 
% \cite{11239499}.
% % and robust team control under uncertainty 
% \cite{10256148}.
%\cite{10476651,10256148,10130652,10018886,9904296,9428540,8340791}.
% Multi-team games provide a natural framework for characterizing the interplay between intra-team cooperation and inter-team competition in distributed decision-making systems \cite{kim2022games, donmez2024team, xiao2025solving}.
% \vspace{-3pt}
Existing studies develop theoretical and algorithmic tools for multi-team interactions from several perspectives.
Feng et al. \citep{feng2021mixed} study a two-team mean-field game that incorporates both within-team cooperation and inter-team competition.
Other studies consider coalition games in which each fixed coalition is treated as a team and analyze the corresponding equilibrium properties \citep{zeng2019generalized,ma2022fully,zhang2023distributed,deng2024distributed}.
Approaches based on mean-field games and coalition games typically take intra-team cooperation as a premise induced by a team-level objective. 
In contrast, team fictitious play (team-FP) focuses on team-level learning dynamics under complete observation of all agents.
Dönmez et al. \citep{donmez2024team} study a multi-team setting in which self-interested agents learn team-level cooperative behavior and propose team-FP, which is shown to converge to a team-Nash equilibrium (TNE).

% \vspace{-3pt}
However, team-FP \citep{donmez2024team} does not address multi-team learning in which self-interested agents learn to exhibit team-level cooperative behavior under incomplete and possibly unreliable information about opponents.
This setting is relevant to applications such as market competition \citep{shi2020multi} and security problems \citep{alcantara2020repeated}, where self-interested agents may not naturally coordinate toward a team-level objective and decisions are often made under incomplete information about other teams.
In addition, the information available for learning can be unreliable or deliberately manipulated, for example, through Byzantine attacks, which may affect agents' decisions and long-term strategic behavior.
The challenge is to establish how self-interested agents can learn team-level cooperative behavior when incomplete information and Byzantine attacks affect their decisions and, consequently, their long-term behavior.

% A key assumption underlying (independent) Team-FP is that each agent can fully observe the actions of all other agents.
% However, learning under limited information has become increasingly important in practice due to privacy and security considerations.
% In scenarios such as defence problems and market competition, agents often do not have direct access to complete information about opponents.
% As a result, agents can no longer update their strategies based on exact observations of all other agents' actions, and the resulting learning dynamics may differ substantially from those under full observation.
% Consequently, (independent) Team-FP cannot be directly applied in such settings.
% The challenge is that imperfect information about opponents affects agents' decisions throughout the learning process, and its effect may accumulate over time, ultimately influencing agents' long-term behaviours.
% This raises the question of how self-interested agents can learn team behaviours under a limited information structure. 
% Motivated by this, in this paper we study multi-team games with self-interested agents under limited information.
% \vspace{-3pt}
To address this challenge, it is important to introduce supervisors as high-level information intermediaries and construct a supervisor network for distributed information learning.
The design of supervisors is motivated by information-intermediary roles studied in organizational management and networked systems \citep{neal2022defining,boutcher2022roles}.
Each supervisor is associated with a subset of teams, receives the joint actions reported by these teams, and provides estimates of team-related information to the agents in these teams.
Supervisors then exchange information over the supervisor network and update their estimates of team-related information for all teams.
The supervisor network provides a basis for distributed learning under incomplete information and Byzantine attacks in multi-team games.

% This supervisor-based framework facilitates multi-team learning under incomplete information and provides a basis for Byzantine-attack identification.
% \vspace{-3pt}
The above motivates the development of algorithms for team orchestration based on distributed information learning over the supervisor network.
For algorithm design, supervisors estimate and exchange information about all teams, and correspondingly, agents choose their actions using the information provided by supervisors.
In addition, algorithmic resilience against Byzantine attacks should be taken into account since such attacks can corrupt supervisors' information learning and affect agents' strategy optimization.
On this basis, we will explore the convergence of the learning dynamics and investigate whether self-interested agents can learn team-level cooperative behavior, as well as the resilience against Byzantine attacks.
% This model captures scenarios in which information about opponent teams is not directly observed but can be inferred through information exchange.
% Consequently, agents update their strategies based on imperfect information conveyed through supervisor networks rather than full observations of all agents' actions.
% Within this framework, we examine whether self-interested agents can learn team behaviour under imperfect information, and analyse its impact on the long-term behaviour and convergence properties of the learning dynamics.
% Building on the above model, we develop a learning framework for multi-team games under limited information.
% We analyse the resulting learning dynamics, establish their convergence properties, and further examine their resilience to Byzantine attacks.

% \vspace{-3pt}
The primary contributions of this paper are as follows:
% \vspace{-6pt}

\begin{itemize}[leftmargin=*]
    \item We design a supervisor-network mechanism for distributed learning of team-related information. Compared with existing studies on multi-team games \citep{feng2021mixed,donmez2024team,deng2024distributed}, the proposed mechanism does not require direct observation of agents in other teams or a priori intra-team cooperation. This enables agents to optimize their strategies using supervisor-mediated belief estimates, thereby supporting team orchestration.
    \item We develop the distributed team-orchestrating algorithm (DTOA) for team-FP learning over the supervisor network. We establish the convergence of the belief-estimation errors over the uniformly connected supervisor network using a gossip-matrix contraction argument (Theorem \ref{belief_convergence_dynamic}). We further derive an upper bound on the team-Nash gap (TNG) by comparing the actual learning dynamics with ideal and reference dynamics (Theorem \ref{convergence_of_gap}).
    \item We investigate a Byzantine attack setting where some teams misreport their actions to supervisors. We propose the Byzantine-resilient DTOA (BR-DTOA) and show the convergence of the belief-estimation errors for honest teams (Theorem \ref{byzantine_belief_errors_convergence}). An upper bound on the honest TNG is revealed, and the near-TNE convergence for honest teams is resiliently preserved against Byzantine attacks with high-probability identification guarantees (Theorem \ref{honest_TNG}).
\end{itemize}

% \begin{comment}
% \vspace{-3pt}    
The rest of this paper is organized as follows.
Section \ref{sec:related_work} reviews related work.
Section \ref{sec:game_formulation} formulates zero-sum potential team games (ZSPTGs) with a supervisor network and their Byzantine extension.
Section \ref{sec:algorithm_1} presents DTOA and establishes convergence and near-TNE guarantees.
Section \ref{sec:algorithm_2} develops BR-DTOA and analyzes its convergence and resilience.
Section \ref{sec:experiments} presents the experimental results.
Section \ref{sec:conclusion} concludes this paper.
The code is available at \url{https://github.com/zjt-1229/team_game_with_supervisor_network}.

%% file: sections_IEEE/2_Related_work.tex
\section{Related Work}
\label{sec:related_work}
This section provides a literature review.

% \vspace{-3pt}
\textbf{Mean-field games (MFG).} Studies on MFGs primarily characterize optimal responses of individuals or representative agents to the mean field in large-scale weakly coupled systems, together with a consistency condition between the induced aggregate statistics and the hypothesized mean field.
In standard MFGs, participants are typically modeled as anonymously coupled individuals through aggregate distributions 
%\citep{huang2006large,lasry2006jeux,lasry2006jeux2}
\citep{huang2006large}.
Subsequent studies further investigate population-level consistency induced by mean-field responses \citep{moll2026mean}, as well as associated learning and computational methods \citep{wu2024population}.
In parallel, the mean-field team literature focuses on team-optimal decision-making under mean-field coupling \citep{sanjari2020optimal} 
%\citep{arabneydi2015team,sanjari2020optimal},
and extends this line to more general uncertainty settings \citep{feng2025mean}.
More recently, mixed cooperative--competitive mean-field models have incorporated both cooperation and competition into the mean-field framework \citep{feng2021mixed,11052671}. Related studies have also used mixed-coalition formulations with intra-group cooperation and inter-group competition \citep{huang2024linear} and large-scale competitive team learning \citep{jeloka2025learning} to characterize richer collective interactions.

% \vspace{-3pt}
Despite these advances, mean-field-based approaches typically characterize agents' responses to aggregate population statistics rather than team-level equilibrium learning among strategically interacting teams.
In contrast, our work addresses the complementary problem of equilibrium formation in repeated multi-team games under limited information, where agents' strategy optimization depends on estimated information acquired through the learning process.

% \vspace{-3pt}
\textbf{Coalition games.} 
Coalitional-game studies take coalitions as the basic units of analysis, focusing either on coalition formation, coalition values, and the stability of payoff allocations, or, under a fixed coalition partition, on equilibria induced by coalition-level objectives.
The early coalition-game literature mainly developed around hedonic preferences \citep{dreze1980hedonic}
% , endogenous coalition formation \citep{hart1983endogenous}, 
and broader modeling and stability analysis of coalition formation \citep{hajdukova2006coalition}.
Recent studies further examine the endogenous grouping of self-interested agents and the structural stability of the resulting partitions from a group-formation perspective \citep{9683571}, and provide a systematic account of major classes of coalitional games and their distributed-network applications \citep{saad2010hedonic}.
%\citep{saad2009coalitional,saad2010hedonic}.
In parallel, the literature on coalition games with fixed coalition partitions treats each coalition as a composite agent and studies the existence and distributed computation of generalized Nash equilibria (GNEs) and variational generalized Nash equilibria (vGNEs), together with extensions to nonsmooth, constrained, and dynamic settings \citep{zeng2019generalized,ma2022fully,zhang2023distributed,deng2024distributed}.
% \citep{zeng2019generalized,ma2022fully,li2023distributed,deng2024distributed,sun2024prescribed,sun2025prescribed}.

% \vspace{-3pt}
However, studies on multi-coalition games with fixed coalition partitions typically model each coalition as a unified agent with a prescribed coalition-level objective and therefore do not capture how team-level behavior emerges from self-interested agent-level interactions.
In contrast, our work studies equilibrium formation among teams induced by agent-level learning, which is relevant to multi-team systems without centralized team control.

\textbf{Byzantine-resilient mechanisms.} 
Uncertainty is an unavoidable issue in many multi-agent game-theoretic settings, and seeking robust or resilient equilibria has become a common approach for preserving desirable performance in uncertain environments \citep{zhu2015game,chen2021distributed,chen2025inverse}.
Among different sources of uncertainty, Byzantine attacks represent a particularly challenging form, in which malicious agents can distort the information available to others by sending falsified or inconsistent reports.
Research on Byzantine-resilient multi-agent systems mainly focuses on how normal agents maintain coordination, identify malicious information sources, and suppress the influence of such sources when adversarial agents disrupt the system by sending falsified messages \citep{wang2023resilient}.
Early work focuses primarily on resilient consensus and establishes convergence mechanisms in the presence of adversarial or Byzantine nodes \citep{leblanc2013resilient}.
%\citep{leblanc2011consensus,leblanc2013resilient}.
Building on this line, subsequent studies introduce mechanisms such as 
% trusted nodes \citep{abbas2014resilient} and 
distributed detection \citep{yuan2021secure}.
These ideas have also been extended to more general control and learning settings, including Byzantine-resilient output regulation \citep{yan2021resilient} and multi-agent reinforcement learning \citep{li2024byzantine}.

% \vspace{-3pt}
These works mainly develop Byzantine-resilient mechanisms for consensus, control, and cooperative learning, with an emphasis on filtering malicious information and preserving coordination among normal agents.
They provide useful insights for the Byzantine-resilient algorithm developed in this paper for multi-team games with self-interested agents.

%% file: sections_IEEE/4_Game_formulation.tex
\section{Problem Formulation}
\label{sec:game_formulation}
% \vspace{-3pt}
In this section, we first revisit ZSPTGs and then formulate ZSPTGs with a supervisor network. 
We also specify the Byzantine attack setting considered in this paper.

\subsection{Revisiting Zero-sum Potential Team Game}
\label{sec:revisiting}
% \vspace{-3pt}
In multi-team games, teams compete against others, while agents in each team cooperate.
A multi-team game is characterized by the tuple $\mathcal{G}\!=\!(\mathcal{I},\mathcal{T}, \{\mathcal{A}^i\}_{i \in \mathcal{I}}, \{ \!u^i \!\}_{i \in \mathcal{I}})$, where $\mathcal{I}$ and $\mathcal{T}$ denote the index sets of agents and teams, respectively.
Let $\mathcal{I}^m$ denote the index set of agents in team $m\! \in\! \mathcal{T}$.
These sets form a partition of $\mathcal{I}$, i.e., $\mathcal{I}^m \!\cap \mathcal{I}^{m'}\!\!\!=\!\! \varnothing$ for $m \!\!\neq\!\! m'$ and $\bigcup_{m \in\! \mathcal{T}} \!\mathcal{I}^m \!\!=\!\! \mathcal{I}$.
Here, $\mathcal{A}^i$ denotes the finite action set of agent $i$, $\mathcal{A} \!\triangleq\! \prod_{i \in \mathcal{I}} \mathcal{A}^i$ denotes the finite joint action set of all agents, and $u^i\!\!:\!\mathcal{A} \!\rightarrow\! \mathbb{R}$ denotes the utility function of agent $i$.
Let $\underline{\mathcal{A}}^m \!\!\triangleq \!\prod_{i \in \mathcal{I}^m}\mathcal{A}^i$ denote the joint action set of team $m \in \mathcal{T}$.
The mixed-strategy spaces of agent $i$, team $m$, and all agents are $\Delta(\mathcal{A}^i)$, $\Delta(\underline{\mathcal{A}}^m)$, and $\Delta(\mathcal{A})$, respectively.
% The function $u^i:A \rightarrow \mathbb{R}$ denotes the utility of agent $i$. 
% We assume that the utility of each agent is finite.
% Thus there exists some $u_{max}>0$ such that $0 \le u^i(a) \le u_{max},\forall i \in \mathcal{I},\forall a \in A$.
We next recall the definition of ZSPTGs, TNG and TNE \citep{donmez2024team}.

\begin{definition}[Zero-sum Potential Team Game]
A multi-team game is a ZSPTG if, for every team $m \in \mathcal{T}$, there exists a potential function $\phi^m: \mathcal{A} \rightarrow \mathbb{R}$ such that
% \vspace{-3em}
\begin{align}
\label{potential_function_difference}
\phi^m(\hat{a}^i\!,a^{-i}\!,\underline{a}^{-m})\!-\!\phi^m(a)\!=\!u^i(\hat{a}^i\!,a^{-i}\!,\underline{a}^{-m})\!-\!u^i(a),
\end{align}
% \vspace{-2em}
for all $(\hat{a}^i,a) \in \mathcal{A}^i \times \mathcal{A}$ and all $i \in \mathcal{I}^m$, where $a^{-i} \triangleq \{a^j\}_{j \in \mathcal{I}^m\setminus\{i\}}$ and $\underline{a}^{-m} \triangleq \{\underline{a}^l\}_{l \in \mathcal{T}\setminus\{m\}}$.
Moreover, the potential functions satisfy the zero-sum condition
% \vspace{-3em}
\begin{align}
    \sum\nolimits_{m \in \mathcal{T}} \phi^m(a)=0,\quad \forall a \in \mathcal{A}.
\end{align}
% \vspace{-2em}
The cross-team interactions are network-separable, so the potential and utility functions can be decomposed as
% \vspace{-3em}
\begin{align}
    \phi^m=\sum\nolimits_{l \neq m} \phi^{ml}, \text{ and } u^i=\sum\nolimits_{l \neq m} u^{il},\ \forall i \in \mathcal{I}^m,
\end{align}
% \vspace{-2em}
for some $\phi^{ml}:\underline{\mathcal{A}}^m \times \underline{\mathcal{A}}^l \rightarrow \mathbb{R}$ and $u^{il}:\underline{\mathcal{A}}^m \times \underline{\mathcal{A}}^l \rightarrow \mathbb{R}$.
\end{definition}

\begin{definition}[Team-Nash Gap]
\label{TNG_definition}
Given a team strategy profile $\pi\!= \!\{ \pi^m\! \in \!\Delta(\underline{\mathcal{A}}^m) \}_{m \in \mathcal{T}}$, the TNG for each team $m \in \mathcal{T}$ is defined as
% \vspace{-3em}
\begin{align*}
TNG^m(\pi) \triangleq  \mathop{\max}\nolimits_{\pi' \in \Delta (\underline{{\mathcal{A}}}^m)} \{\phi^m(\pi',\pi^{-m})\}-\phi^m(\pi).
\end{align*}
% \vspace{-2em}
The TNG is then defined as $TNG(\pi) \!\!\triangleq\!\! \sum_{m \in \mathcal{T}} \!TNG^m(\pi)$, where $\pi^{-m} = \{ \pi^l \}_{l \in \mathcal{T}\setminus{m}}$.
Correspondingly, $\pi$ is called a TNE if $TNG(\pi)=0$.
Furthermore, for any $\varepsilon>0$, $\pi$ is called an $\varepsilon$-TNE if $TNG(\pi) < \varepsilon$.
\end{definition}

% \vspace{-3pt}
Intuitively, a small TNG means that the induced team strategy profile is nearly stable against unilateral deviations of any single team from its current strategy.
When the team strategy profile has a small TNG, the resulting behavior reflects an approximate pattern of within-team cooperation and inter-team competition.
Thus, we use the TNG to evaluate the extent to which self-interested agents in each team behave cooperatively while competing with other teams.

% \vspace{-3pt}
To study how self-interested agents maximize their own utilities, we next revisit the smoothed best response used in the subsequent learning algorithms.

% \textcolor{blue}{
% \begin{definition}[Social Welfare]
%     Given a policy $\pi$, we define the social welfare of $\pi$ as
%     \begin{align*}
%         W(\pi)=\sum_{i \in \mathcal{I}} u^i(\pi). 
%     \end{align*}
%     Furthermore, $\pi^*$ is a optimal policy if $W(\pi^*)=\mathop{max}_{\pi \in \Pi_{m \in \mathcal{T}}\Delta(\underline{A}^m)} W(\pi) \triangleq W^*$.
% \end{definition}
% }

\begin{definition}[Smoothed Best Response]
\label{smoothed_best_response}
Given a finite set $X$, for a distribution $D \!\!\in\!\! \Delta(\!X\!)$ and a function $f\!:\!X \!\!\rightarrow\!\! \mathbb{R}$, let $f(D)\!\!=\!\!\sum_{x\in \!X}\!D(x)f(x)$.
For a temperature parameter $\tau>0$, the smoothed best response to $f$ is defined as
% \vspace{-3em}
\begin{align*}
% \label{smmothed_best_response}
    br_{\tau}(f)(x)=\frac{e^{\frac{f(x)}{\tau}}}{\sum_{\tilde{x} \in X} e^{\frac{f(\tilde{x})}{\tau}}},\ \forall x \in X.
\end{align*}
% \vspace{-2em}
Equivalently, $br_{\tau}(f)\!\!=\!\!\mathop{\arg\max}_{D \in \Delta(\!X\!)}\!\! \left\{ f(D)\!+\!\tau \mathcal{H}(D) \right \}$ with entropy regularization $\mathcal{H}(D) \!\!\triangleq\!\! \sum_{x\in\! X}\! -\!D(\!x\!)\!\log\! D(\!x\!) $.
\end{definition}

% \vspace{-3pt}
The smoothed best response provides a regularized decision rule that assigns positive probability to all feasible actions and balances objective optimization with exploration.
The temperature parameter $\tau$ controls the degree of smoothing: smaller values make the response closer to a pure best response, whereas larger values induce more exploratory behavior.
The maximizer is unique because $f$ is linear and, for $\tau>0$, the entropy regularization term makes the objective strictly concave on $\Delta(X)$.

% \vspace{-3pt}
Note that the utilities of all agents need not sum to zero.
Although the definition of a ZSPTG requires the team potentials to sum to zero, the same analysis applies when their sum is a constant independent of the action profile $a \!\in \!\mathcal{A}$ and agents choose actions according to Definition \ref{smoothed_best_response}.
Indeed, the potentials can be normalized by subtracting this constant from any one team potential, which preserves all potential differences and therefore leaves the smoothed best-response structure unchanged.

\subsection{ZSPTG with a supervisor network}
% \vspace{-3pt} 
In this paper, we focus on ZSPTGs with a supervisor network, where supervisors serve as distributed information intermediaries.
A multi-team game with a supervisor network is characterized by the tuple $\mathcal{G}=(\mathcal{I},\mathcal{T}, \mathcal{S},ST,\{SN_k\}_{k \ge 0},\{\mathcal{A}^i\}_{i \in \mathcal{I}}, \{ u^i \}_{i \in \mathcal{I}})$, where $\mathcal{I}$, $\mathcal{T}$, and $\mathcal{S}$ denote the index sets of agents, teams, and supervisors, respectively.
We follow the notation in Section \ref{sec:revisiting}: $\mathcal I^m$, $\mathcal{A}^i$, $\underline{\mathcal{A}}^m$, $\mathcal{A}$, and $u^i:\mathcal{A}\to\mathbb R$ denote the index set of agents in team $m$, the action set of agent $i$, the joint action set of team $m$, the joint action set of all agents, and the utility function of agent $i$, respectively.
We assume that each agent observes only the actions of agents in the same team.
The matrix $ST\in \{ 0,1 \}^{\vert \mathcal{T} \vert \times \vert \mathcal{S} \vert}$ denotes the supervision relationships between teams and supervisors, and the matrix $SN_k\in \{ 0,1 \}^{\vert \mathcal{S} \vert \times \vert \mathcal{S} \vert}$ denotes the communication structure among supervisors at round $k$.

% \vspace{-3pt}
Each supervisor supervises a subset of teams: it receives action reports from these teams and provides their agents with information about other teams.
Specifically, $ST(m,s)\!\!=\!\!1$ indicates that team $m$ is supervised by supervisor $s$, and $ST(m,s)\!\!=\!\!0$ otherwise.
Similarly, $SN_k(s,s')=1$ indicates that supervisors $s$ and $s'$ exchange information at round $k$, and $SN_k(s,s')=0$ otherwise.
Each agent aims to maximize its utility based on the information provided by the supervisors.
For each supervisor $s \!\in\! \mathcal{S}$, let $\mathcal{T}^s \!\subseteq\! \mathcal{T}$ denote the set of teams supervised by $s$.
The supervisor communication structure at round $k$ can be equivalently represented by an undirected graph $G_k\!\!=\!\!(\mathcal{S},\!E_k)$, where $\mathcal{S}$ is the node set and $E_k\!\!=\!\!\big\{\! \{s,s'\}\!:\! SN_k(s,s')\!\!=\!\!1 \!\big\}$ is the edge set, for all $k\!\ge\!0$.

% We make the following assumptions in this paper.
% \begin{assumption}
% \label{team_supervisor_assumption}
%     Every team is supervised by at least one supervisor, i.e., $\cup_{s \in \mathcal{S}} \mathcal{T}^s=\mathcal{T}$.
% \end{assumption}

% % \begin{assumption}
% % \label{supervisor_network_assumption}
% % $G$ is a undirected connected graph.
% % \end{assumption}

% \begin{assumption}
% \label{dynamic_assumption}
%     The supervisor network $SN_k=G_k \triangleq (\mathcal{S},E_k)$ changes and is uniformly connected over time \citep{moreau2005stability}. 
%     Specifically, let $B$ be the smallest positive integer such that, for any $n \in \mathbb{N}_+$, the graph $G_{(n)} \triangleq(\mathcal{S},\cup_{k=nB}^{(n+1)B-1}E_k)$ is connected.
% \end{assumption}

% Assumption \ref{team_supervisor_assumption} indicates that each team's information is exposed to the supervisor network via at least one supervisor.
% Assumption \ref{dynamic_assumption} ensures that any two supervisors can communicate with each other indirectly within limited time steps (Lemma \ref{dynamic_path_lemma}).

% In Team-FP-type learning dynamics, agents update their actions based on beliefs about other teams' strategies \citep{donmez2024team}; hence, inaccurate belief information may distort action updates and affect the resulting team behaviours.
% In ZSPTGs with supervisor networks, such inaccuracy may arise from supervisors' belief-estimation errors.
% Therefore, it is crucial to establish whether the learning dynamics still converge to a near TNE when agents rely on inaccurate belief information.
% \vspace{-5pt}
In the team-FP learning dynamics, agents update their actions based on common beliefs about other teams' strategies \citep{donmez2024team}, which are formed through direct observation of the actions of other agents.  
In ZSPTGs with a supervisor network, however, these beliefs are provided by supervisors that receive action reports from supervised teams and exchange information over the supervisor network.
Consequently, belief-estimation errors can arise because each supervisor receives reports from only a subset of teams.
It is therefore crucial to characterize how such errors affect the convergence of the learning dynamics and the induced TNG.
This motivates the following problem.

\begin{problem}
\label{problem_1}
    How to design a distributed team-orchestrating algorithm for ZSPTGs with a supervisor network while providing theoretical guarantees? 
\end{problem}
% \textit{Problem 1:}
% How can (independent) Team-FP be extended to ZSPTGs with supervisor networks while preserving convergence to a near TNE under belief-estimation errors?

% \vspace{-3pt}
This problem concerns whether the common-belief requirement in team-FP can be relaxed through distributed belief learning over the supervisor network. 
The key challenge is to control the long-term effect of supervisors' belief-estimation errors on the induced team behavior.
We first characterize the belief-estimation process over the supervisor network and then quantify its effect on the induced learning dynamics and the TNG.
We impose the following assumptions.
% \vspace{3pt} 
\begin{assumption}
\label{team_supervisor_assumption}
    Every team is supervised by at least one supervisor, i.e., $\bigcup_{s \in \mathcal{S}} \mathcal{T}^s=\mathcal{T}$.
\end{assumption}

\begin{assumption}
\label{dynamic_assumption}
    The supervisor communication graph $G_k \triangleq (\mathcal{S},E_k)$, which represents the supervisor network, is time-varying and uniformly connected over time \citep{moreau2005stability}. 
    Specifically, let $B$ be the smallest positive integer such that, for any $n \in \mathbb{N}^+$, the union graph $G_{(n)} \triangleq(\mathcal{S},\bigcup_{k=nB}^{(n+1)B-1}E_k)$ is connected.
\end{assumption}

% \vspace{-3pt}
Assumption \ref{team_supervisor_assumption} ensures that every team has at least one reporting link to the supervisor network and hence no team is completely isolated from it.
Assumption \ref{dynamic_assumption} ensures that information can propagate among all supervisors within each finite communication window.

\subsection{Byzantine Attack}
% \vspace{-3pt}
We consider a Byzantine attack setting in which some teams can misreport their actions to the supervisors at each round.
Given a ZSPTG with a supervisor network $\mathcal{G}=(\mathcal{I},\mathcal{T},\mathcal{S},ST,\{SN_k\}_{k\ge 0},\{ \mathcal{A}^i \}_{i \in \mathcal{I}},\{ u^i \}_{i \in \mathcal{I}})$, honest teams and Byzantine teams are defined as follows.

\begin{definition}[Honest Team and Byzantine Team]
    A team $m$ is called honest if it reports its joint action truthfully to the supervisors that supervise it, i.e., $\underline{a}_{k,r}^m=\underline{a}_k^m$ for all $k \ge 0$, where $\underline{a}_{k,r}^m$ denotes the joint action reported by team $m$ at round $k$.
    In contrast, a team $m$ is called Byzantine if it can misreport its joint action to the supervisors that supervise it, i.e., $\underline{a}_{k,r}^m \neq \underline{a}_k^m$ for some $k \ge 0$.
\end{definition}
% \vspace{-3pt}
Let $\mathcal{H} \subseteq \mathcal{T}$ and $\mathcal{B} \subseteq \mathcal{T}$ denote the sets of honest and Byzantine teams, respectively.
% We next define the honest TNG.

% Based on the definitions of honest and Byzantine teams, we define the honest team-Nash gap, referred to as honest TNG.
\begin{definition}[Honest Team-Nash Gap]
    Given a team strategy profile $\pi= \{ \pi^m \in \Delta(\underline{\mathcal{A}}^m) \}_{m \in \mathcal{T}}$, the honest TNG is defined as 
    $TNG_{\mathcal{H}}(\pi) \triangleq \sum_{m \in \mathcal{H}} TNG^m(\pi)$.
    Correspondingly, the team strategy profile $\pi$ is called an $\varepsilon$-honest TNE if $TNG_{\mathcal{H}}(\pi)<\varepsilon$.
\end{definition}
% \vspace{-3pt}
In the definition of the honest TNG, the strategies of Byzantine teams are treated as fixed exogenous factors, and the gap is evaluated only over honest teams.
Since our objective is to optimize the performance of honest teams under Byzantine attacks, including the gaps for Byzantine teams in the honest TNG does not reflect the intended performance measure.

% In this Byzantine attacking setting, the learning dynamics is exposed to Byzantine teams that may misreport their actions to supervisors. 
% These misreports will distort the learning process and influence the long-term behaviours of agents within honest teams. 
% This raises the following resilient learning problem.
% \vspace{-3pt}
In this Byzantine attack setting, the learning dynamics are affected by Byzantine teams, whose misreports can distort supervisors' belief learning processes and influence the long-term behavior of agents in honest teams. 
This leads to the following Byzantine-resilient learning problem.
\begin{problem}
\label{problem_2}
    How to design a distributed resilient algorithm for ZSPTGs with a supervisor network against Byzantine attacks?
\end{problem}

% \vspace{-3pt}
The key challenge is to identify Byzantine teams using limited information while preventing their misreports from disrupting the long-term learning behavior of honest teams.
We address this challenge by developing a supervisor-based mechanism and quantifying the effect of Byzantine teams on the learning dynamics of honest teams.
To this end, each supervisor can check whether the action reports from its supervised teams are consistent with their true actions. 
The checking outcome is not necessarily accurate: an honest report can be incorrectly regarded as a misreport, and a misreport can fail to be detected.
% To this end, supervisors can check the action reports from the teams under their supervision, but the checking results may be imperfect.
% This checking process is specified by the following assumption.
% \begin{assumption}
% \label{assumption:imperfect_checking}
% Each supervisor is able to check whether the action reports from its supervised teams are consistent with their true actions.  
% The checking result is not necessarily accurate: an honest report may be incorrectly regarded as a misreport, and a misreport may fail to be detected.
% \end{assumption}

% The introduction of supervisor networks gives rise to belief-estimation errors, so teams can no longer have access to the accurate common beliefs assumed in \citet{donmez2024team}.
% Such estimation errors affect agents' strategy updates and propagate over rounds, thereby influencing the evolution of team behaviours.
% We therefore adapt (Independent) Team-FP \citep{donmez2024team} to multi-team games with supervisor networks and analyse the convergence of belief learning and the team-Nash gap under belief-estimation errors. 
% Furthermore, we incorporate Byzantine teams into the proposed framework and design a resilient mechanism to mitigate their impact on the learning process.

%% file: sections_IEEE/5_Convergence.tex
\section{Distributed Team-Orchestrating Algorithm}
\label{sec:algorithm_1}
% \vspace{-3pt}
This section addresses Problem \ref{problem_1} by developing the DTOA.
We then analyze the convergence of supervisors' belief-estimation errors and derive an upper bound on the TNG.

%We first present the TOA, which combines team fictitious play with decentralised belief learning over supervisor networks.
%We then establish the convergence of supervisors' belief estimates.
%Finally, we show that the learning process under TOA converges to a near TNE in terms of the TNG.

\subsection{Algorithm Design}

\begin{algorithm}[t]
    \caption{\!Distributed \!Team-Orchestrating\! Algorithm (DTOA)}
    \label{alg:team_fp_with_decentralized_belief_learning}
    \begin{algorithmic}[1] 
        \STATE  \textbf{Initialize:} $\{ \pi_{s,0}^m \}_{m \in \mathcal{T}}$ and $\{a_{-1}^i \}_{i \in \mathcal{I}}$ arbitrarily
        % , $\{ \lambda^m_0=\lambda_{min} \}_{m \in \mathcal{T}}$

        \WHILE{round $k = 0, 1, \dots$}
            \FOR{$m \in \mathcal{T}$}
                \STATE select agent $i \in \mathcal{I}^m$ in team $m$ randomly 
                \STATE select $s \in \left\{s : i \in \mathcal{T}^s \right\}$ randomly to provide belief estimates about other teams
                \STATE agent $i$ updates its action based on $a_{k-1}^{-i}$ and $\pi_k^{-m,s}$: $a_k^i \sim br_{\tau} \left( u^i(\cdot,a_{k-1}^{-i},\pi_k^{m,s}) \right)$
                \FOR {agent $j \in \mathcal{I}^m \backslash \{ i \}$}
                    \STATE repeat the last action: $a_k^j=a_{k-1}^j$
                    % \STATE $\lambda_{k+1}^m=\lambda_{k}^m$
                \ENDFOR
                % \STATE update $\lambda^m_{k+1}=\mathop{max}\left\{\lambda^m_k+\frac{1}{t_{\lambda,k}} c^m\left (\underline{a}_{k}^m \right),\lambda_{min} \right\}$
            \ENDFOR
            \FOR{$s \in \mathcal{S}$ and $m \in \mathcal{T}$}

                \IF{$m \in \mathcal{T}^s$}
                    \STATE update: $\pi_{k+1}^{m,s}=\pi_{k}^{m,s}+\alpha_k \left( \underline{a}_k^m -\pi_{k}^{m,s} \right)$
                \ELSE
                    \STATE update: $\pi_{k\!+\!1}^{m,s}\!\!=\!\!
                    \begin{cases}
                        \!\!\frac{1}{\vert \mathcal{N}^s_k \vert}\!\!\sum_{s'\! \in \mathcal{N}^s_k} \!\pi_{k}^{m,s'}\!\!\!, \!\! &\text{for }\vert \mathcal{N}^s_k \vert\!\!>\!\!0\\
                        \!\pi_{k}^{m,s}, \!\!  &\text{for }\vert \mathcal{N}^s_k \vert\!\!=\!\!0
                    \end{cases}$
                    
                \ENDIF
            \ENDFOR

        \ENDWHILE
    \end{algorithmic}
\end{algorithm}

% \vspace{-3pt}
We first specify the main components of DTOA for ZSPTGs with a supervisor network.
At round $k$, let $a_k^i$ denote the action of agent $i$, and let $\underline{a}_k^m=\left(a_k^i\right)_{i \in \mathcal{I}^m}$ denote the joint action of team $m$.
To compensate for the lack of complete opponent-action information, DTOA uses supervisors to provide belief estimates.
Let $\pi_k^{m,s} \in \Delta\left(\underline{\mathcal{A}}^m\right)$ denote the belief of supervisor $s$ regarding the strategy of team $m$.
To induce team-level strategic behavior under incomplete opponent-action information, agent $i$ in team $m$ updates its action according to the smoothed best response in Definition \ref{smoothed_best_response}, using the previous actions of the other agents in the same team and the beliefs provided by a supervisor $s$ that supervises team $m$:
% \vspace{-3em}
\begin{align}
\label{agent_action_update}
a_k^i \sim br_{\tau} \left( u^i(\cdot,a_{k-1}^{-i},\pi_k^{-m,s}) \right),
\end{align}
% \vspace{-2em}
where $a_k^i \sim p$ means that $a_k^i$ is sampled from $p \in \Delta\left(\mathcal{A}^i\right)$.

% \vspace{-3pt}
The rule above specifies the action update of a selected agent.
In DTOA, one agent per team is randomly selected at each round to update its action, while the other agents in the same team keep their actions unchanged.
When such coordination is unavailable, we consider an independent variant, called independent DTOA (iDTOA), in which each agent updates its action independently with probability $\delta \in (0,1)$.

% \vspace{-3pt}
To maintain belief estimates of teams' strategies, we design two belief-update rules: one for teams directly supervised by a given supervisor and the other for teams not supervised by that supervisor.
Given a supervisor $s \in \mathcal{S}$, for a team $m \in \mathcal{T}^s$, supervisor $s$ directly updates its belief using the joint action reported by team $m$:
% \vspace{-3em}
\begin{align*}      
    \pi_{k+1}^{m,s}=\pi_k^{m,s}+\alpha_k \left( \underline{a}_k^{m} -\pi_k^{m,s} \right),
\end{align*}
% \vspace{-2em}
where $\left\{\! \alpha_k \!\right\}_{\!k \ge 0}$ is the step-size sequence used by all supervisors, and $\underline{a}_k^m$ denotes the one-hot representation in $\Delta\left(\underline{\mathcal{A}}^m\right)$ for notational simplicity.
Let $\mathcal{N}^s_k\!\!= \!\!\{\! s' \!\!\in\!\! \mathcal{S}\!\!:\! SN_k(\!s,\!s'\!)\!\!=\!\!1 \!\}$ denote the neighbor set of supervisor $s$ at round $k$.
For a team $m \!\!\notin\!\! \mathcal{T}^s$, if $\vert \mathcal{N}^s_k \vert \!\!>\!0$, supervisor $s$ updates its beliefs using information from its neighbors:
% \vspace{-3em}
\begin{align*}
    \pi_{k+1}^{m,s}=\frac{1}{\vert \mathcal{N}^s_k \vert}\sum_{s' \in \mathcal{N}^s_k} \pi_{k}^{m,s'}.
\end{align*}
% \vspace{-2em}
If $\vert \mathcal{N}^s_k \vert =0$, supervisor $s$ keeps its beliefs unchanged, i.e., $\pi_{k+1}^{m,s}=\pi_k^{m,s}$ for all $m \notin \mathcal{T}^s$.
% where $\mathcal{N}^s_k= \{ s' \in \mathcal{S}: SN_k(s,s')=1 \}$ denotes the set of neighbours of supervisor $s$ at round $k$.
The implementation of DTOA is summarized in Algorithm \ref{alg:team_fp_with_decentralized_belief_learning}.

\subsection{Convergence Analysis}
% \vspace{-3pt}
We now turn to the convergence analysis of DTOA. 
Following the step-size conditions considered in \citet{donmez2024team}, we impose the following assumption.
\begin{assumption}
\label{step_size_assumption}
    The step-size sequence $\{\alpha_k\}_{k \ge 0}$ satisfies the following conditions: 

    % \vspace{-1em}
    \begin{itemize}
        \item $\alpha_k\in [0,1]$, and $\alpha_k \rightarrow 0$ as $k \rightarrow \infty$;
        \item $\sum_{k=0}^{\infty} \alpha_k=\infty$ and $\sum_{k=0}^{\infty} \alpha_k^2 < \infty$;
        \item $\lim_{k \to \infty} {\alpha_k}/{\alpha_{k+1}}=1$ and $\alpha_k-\alpha_{k+1} \ge \alpha_k \alpha_{k+1}$.
    \end{itemize}
\end{assumption}

% \vspace{-1em}
The first two conditions in Assumption \ref{step_size_assumption} ensure that the belief updates continue to incorporate new action information while the effect of sampling fluctuations is asymptotically averaged out.
The last condition further ensures that action information from adjacent periods has comparable influence on the beliefs.
A standard choice satisfying these conditions is $\alpha_k\!\!=\!\!{1}/{(k\!+\!1)}$, which corresponds to empirical averaging over past actions.

% \vspace{-3pt}
To quantify supervisors' belief-estimation errors, we define the true belief updates in the full-supervision case where each supervisor supervises all teams as
% \vspace{-3em}
\begin{align}
\label{belief_update}
    \pi_{k+1}^m=\pi_k^m+\alpha_k \left( \underline{a}_k^m -\pi_k^m \right),\ \forall m \in \mathcal{T}.
\end{align}
% \vspace{-2em}
Under Assumption \ref{step_size_assumption}, the effect of the initial beliefs vanishes asymptotically. 
Hence, the initialization does not affect the asymptotic convergence results. 
For ease of analysis, we set $\pi_0^{m,s}=\pi_0^m$ for all $m\in\mathcal{T}$ and $s\in\mathcal{S}$.

\begin{remark}
    In the full-supervision case, for any supervisor $s\!\in\!\mathcal{S}$ and any team $m\!\in\!\mathcal{T}$, we have $\pi_k^{m,s}\!=\!\pi_k^{m}$ for all $k\ge 0$.
    By contrast, under a general supervision structure, a supervisor $s\in\mathcal{S}$ may not directly supervise a team $m\!\in\!\mathcal{T}$, i.e., $m\!\notin\!\mathcal{T}^s$, in which case the equality $\pi_k^{m,s}\!=\!\pi_k^{m}$ may no longer hold.
    % This motivates the analysis of belief-estimation errors under general supervisor networks.
\end{remark}

% \vspace{-3pt}
For supervisor $s\!\in\!\mathcal{S}$ and team $m\!\in\!\mathcal{T}$, we define the belief-estimation error with respect to the true belief as $\Vert \pi_k^{m,s}\!-\!\pi_k^m \!\Vert_{\infty}$.
The following theorem establishes the asymptotic convergence of all supervisors' belief-estimation errors under DTOA.
The proof is provided in Appendix \ref{proof_of_believe_convergence}.

% \begin{enumerate}
%     \item Bound on CTNG? (Similar framework in \citep{donmez2024team})
%     \item Bound on CV?
% \end{enumerate}

% \begin{theorem}
% \label{belief_convergence}
% Given a ZSPTG with supervisor network $\mathcal{G}=(\mathcal{I},\mathcal{T},\mathcal{S},ST,SN,\{ A^i \}_{i \in \mathcal{I}},\{ u^i \}_{i \in \mathcal{I}})$, let supervisors follow Algorithm \ref{alg:decentralized_belief_learning}.
% If Assumption \ref{team_supervisor_assumption}, \ref{supervisor_network_assumption} and \ref{step_size_assumption} hold, the supervisors' beliefs about all team will converge to the true belief, i.e., 
% \begin{align}
%     \lim_{k \to \infty} \Vert \pi_k^{m,s}-\pi_k^m \Vert_1=0,\ \forall s \in \mathcal{S},\ \forall m \in \mathcal{T}.
% \end{align}
% If we set $\alpha_k=\frac{1}{k}$, then the error of supervisors' beliefs satisfy $\Vert \pi_k^{m,s}-\pi_k^m \Vert_1 \le O(\frac{1}{k}),\ \forall s \in \mathcal{S},\ \forall m \in \mathcal{T}$.
% \end{theorem}

% \subsection{Convergence}
\begin{theorem}
\label{belief_convergence_dynamic}
For a ZSPTG with a supervisor network 
% \vspace{-3em}
\begin{align*}
    \mathcal{G}=(\mathcal{I},\mathcal{T},\mathcal{S},ST,\{SN_k\}_{k\ge 0},\{ \mathcal{A}^i \}_{i \in \mathcal{I}},\{ u^i \}_{i \in \mathcal{I}}),
\end{align*}
% \vspace{-2em}
under Assumptions \ref{team_supervisor_assumption} and \ref{dynamic_assumption}, the supervisors' belief-estimation errors in DTOA converge to zero; specifically,
% \vspace{-3em}
\begin{align}
\label{theorem_1_statement}
    \Vert \pi_k^{m,s}\!-\!\pi_k^m \Vert_{\infty}\!\le\! O(\rho^{[\frac{k}{2}]})\!+\!O(\alpha_{[\frac{k}{2}]}),\forall s \!\in\! \mathcal{S}, \forall m \!\in \!\mathcal{T},
\end{align}
% \vspace{-2em}
where $\rho\!=\!\left(\!1\!-\!\left(\frac{1}{\vert \mathcal{S} \vert}\right)^{(\vert \mathcal{S} \vert+1)B}\right)^{\frac{1}{(\vert \mathcal{S} \vert+1)B}}\!<\!1$.
If Assumption \ref{step_size_assumption} also holds, (\ref{theorem_1_statement}) simplifies to $\Vert \pi_k^{m,s}\!-\!\pi_k^m \Vert_{\infty}\!\le\! O(\alpha_{[\frac{k}{2}]})$.
\end{theorem}

% \vspace{-3pt}
Theorem \ref{belief_convergence_dynamic} shows that the convergence rate of the supervisors' beliefs depends on the supervisor network $\{SN_k\}_{k \ge 0}$, the number of supervisors $\vert \mathcal{S} \vert$, and the step-size sequence $\{\alpha_k\}_{k\ge0}$.
Here, $B$ characterizes the length of the communication window over which information is propagated among supervisors.
The first term in \eqref{theorem_1_statement} indicates that, for fixed $B$ and $\{\alpha_k\}_{k\ge0}$, a larger number of supervisors $\vert \mathcal{S} \vert$ leads to slower convergence.
We next fix $\vert \mathcal{S} \vert$ and $\{\alpha_k\}_{k\ge0}$ and examine how $B$ affects the convergence rate.
Since the logarithm is strictly increasing, the monotonicity of $\rho$ with respect to $B$ is equivalent to the monotonicity of $\log \rho$ with respect to $B$.
For all $\vert \mathcal{S} \vert >1$ and all $B>0$, we have
% \vspace{-3em}
\begin{align*}
    \frac{\partial (\log \rho)}{\partial B}= \frac{(\vert \mathcal{S} \vert+1)B \cdot \frac{\log\vert \mathcal{S}\vert}{\vert \mathcal{S}\vert^{(\vert \mathcal{S} \vert+1)B}-1}}{(\vert \mathcal{S} \vert+1)^2B^2}>0.
\end{align*}
% \vspace{-2em}
This implies that $\rho$ increases with $B$; hence, a larger $B$ leads to a slower convergence rate.
The second term in \eqref{theorem_1_statement} further shows that the step-size sequence $\{\!\alpha_k\!\}_{\!k\ge0}$ also affects the convergence rate of the supervisors' beliefs.

\subsection{TNG Analysis}
% \vspace{-3pt}
Having established the convergence of supervisors' belief estimates, we next analyze TNG convergence under DTOA.
The main challenge in proving TNG convergence for DTOA is that supervisors' belief-estimation errors affect agents' actions at each round, thereby influencing subsequent belief updates and action decisions.
Although supervisors' belief-estimation errors converge to zero as $k \rightarrow \infty$, it remains nontrivial to show that their cumulative influence on agents' long-term behavior also vanishes asymptotically.
% This distinguishes our analysis from the setting without belief-estimation errors in \citep{donmez2024team}.

% \vspace{-3pt}
To address this difficulty, we introduce an ideal scenario for the analysis of DTOA and refer to the original setting with belief-estimation errors as the actual scenario.
The repeated play is divided into epochs, each consisting of $T$ rounds.
In the ideal scenario, at the beginning of each epoch, all supervisors are initialized with the history of all agents' actions from the actual scenario and supervise all teams throughout the epoch. Consequently, they share common beliefs at each round, which are referred to as ideal beliefs.
% The ideal scenario can be viewed as a special case of the actual scenario in which each supervisor supervises all teams.
Since the actual and ideal scenarios use different supervision structures during epoch $n$, they generally induce different distributions over the joint actions of all teams at round $k$ of epoch $n$.
\begin{remark}
\label{true_and_ideal_beliefs_remark}
    The ideal beliefs generally differ from the true beliefs.
    Specifically, the true beliefs depend on the action history in the actual scenario, whereas the ideal beliefs depend on both the actions during the current epoch in the ideal scenario and the action history generated before epoch $n$ in the actual scenario.
    Nevertheless, at the beginning of each epoch, the ideal beliefs coincide with the true beliefs because both are constructed from the action history generated before epoch $n$ in the actual scenario.
\end{remark}

% \vspace{-3pt}
The key idea is to use an epoch-wise comparison between the ideal and actual scenarios to quantify how supervisors' belief-estimation errors propagate to the induced action distributions.
Let $\nu_{k,(n)}^m$ and $a\text{-}\nu_{k,(n)}^m$ denote the joint-action distributions of team $m$ at round $k$ of epoch $n$ in the ideal and actual scenarios, respectively.
The following lemma quantifies how the difference between $\nu_{k,(n)}^m$ and $a\text{-}\nu_{k,(n)}^m$ is bounded in terms of supervisors' belief-estimation errors.

\begin{lemma}
\label{distribution_error}
For a ZSPTG with a supervisor network $\mathcal{G}=(\mathcal{I},\mathcal{T},\mathcal{S},ST,\{SN_k\}_{k\ge0},\{ \mathcal{A}^i \}_{i \in \mathcal{I}},\{ u^i \}_{i \in \mathcal{I}})$, under Assumptions \ref{team_supervisor_assumption} and \ref{dynamic_assumption}, the difference between the induced action distributions for DTOA in the actual and ideal scenarios can be bounded as
% \vspace{-3em}
\begin{align*}
    \Vert \nu_{k,(n)}^m-a \text{-} \nu_{k,(n)}^m \Vert_1 \le C \delta_{(n)},
\end{align*}
% \vspace{-2em}
where $\delta_{(n)}\!\!=\!\!\mathop{\max}_{k \in [nT\!-\!1,(\!n\!+\!1\!)T\!-\!1]}\!\sum_{m \in \mathcal{T}} \!\mathop{\max}_{s \in \mathcal{S}} \!\Vert \pi_{k}^{m,s}\!-\!\pi_{k}^m \Vert_1$ denotes the maximum aggregate supervisor belief-estimation error in epoch $n$.
Since Theorem \ref{belief_convergence_dynamic} implies $\delta_{(n)} \!\!\rightarrow\! 0$ as $n \!\!\rightarrow\! \infty$, it follows that $\left\Vert \nu_{k,(n)}^m\!\!-\!a\text{-}\nu_{k,(n)}^m \!\right\Vert_1\!\!\rightarrow\! 0$.
\end{lemma}

% \vspace{-3pt}
Lemma \ref{distribution_error} shows that the difference between the action distributions induced by the actual and ideal beliefs is bounded by a term proportional to the maximum aggregate supervisor belief-estimation error in epoch $n$.
This is crucial for analyzing agents' long-term behavior in the actual scenario.
The proof is provided in Appendix \ref{proof_of_gap}.
\begin{remark}
    At round $k\!=\!nT$, Remark \ref{true_and_ideal_beliefs_remark} and the properties of the smoothed best response yield that $\Vert \nu_{nT,(n)}^m\!-\!a\text{-}\nu_{nT,(n)}^m \Vert_1$ can be controlled by $\Vert \pi_{nT}^{m,s}\!\!-\!\!\pi_{nT}^m \Vert_1$.
    However, this argument does not directly extend to $nT\!\!\!<\!\!k\!\!\le\!\! (\!n\!+\!1\!)T\!\!-\!1$, since the action distributions are affected by the accumulated discrepancy between the actual and ideal scenarios.
\end{remark}

% \vspace{-3pt}
We now present the TNG result.
Building on Theorem \ref{belief_convergence_dynamic} and Lemma \ref{distribution_error}, the following theorem establishes an almost-sure upper bound on the TNG under DTOA.

\begin{theorem}
\label{convergence_of_gap}
For a ZSPTG with a supervisor network $\mathcal{G}=(\mathcal{I},\mathcal{T},\mathcal{S},ST,\{SN_k\}_{k\ge0},\{ \mathcal{A}^i \}_{i \in \mathcal{I}},\{ u^i \}_{i \in \mathcal{I}})$, under Assumptions \ref{team_supervisor_assumption}, \ref{dynamic_assumption}, and \ref{step_size_assumption}, DTOA satisfies, $\forall s \in \mathcal{S}$,
% \vspace{-3em}
\begin{align*}
     \limsup_{k \to \infty}T\!N\!G(\pi_k^s)
     \!\!\le\!\!  \begin{cases}
        \!\tau\! \log \!\vert \mathcal{A} \vert, \!\!\!\!\!&\text{for DTOA},\\
        \!\tau\! \log \!\vert \mathcal{A} \vert\!+\!\vert \mathcal{T} \vert^2 \overline{\phi}  \Lambda(\delta,\varepsilon_{\phi}), \!\!\!&\text{for iDTOA},
    \end{cases}
\end{align*}
% \vspace{-2em}
where $\overline{\phi}\!\!=\! \!\mathop{\max}_{(m,l,a)} \vert \phi^{ml}(a) \vert$, $\pi_k^s\!\!=\!\!\{ \pi_k^{m,s} \}_{m \in \mathcal{T}}$, and $\Lambda(\delta,\varepsilon_{\phi})$ is a function decaying to zero as $\delta \rightarrow 0^+$ for any $0<\varepsilon_{\phi} \le \min_{a \in \mathcal{A}} br_{\tau} \left( a , a^{-i} \right)$.
\end{theorem}

% \vspace{-3pt}
To proceed with the TNG convergence analysis, we introduce a reference scenario following \citep{donmez2024team}.
In this reference scenario, at any round $k$ in epoch $n$, agents update their actions using the true beliefs at the beginning of epoch $n$.
Specifically, for $nT \!\le\! k \!\le\! (n\!+\!1)T\!-\!1$, the selected agent $i$ follows the reference update rule $a_k^i \!\sim \!br_{\tau} \left( u^i(\cdot,a_{k-1}^{-i},\pi_{nT}^{-m}) \right)$.
The comparison between the actual and reference scenarios can then be decomposed into two parts: the comparison between the actual and ideal scenarios and the comparison between the ideal and reference scenarios.
The proof is in Appendix \ref{proof_of_gap}.

% \vspace{-3pt}
Theorem \ref{convergence_of_gap} shows that the long-term behavior of all agents induces a near TNE.
This implies that agents in the same team learn to cooperate with one another while competing against other teams, even though each agent only knows its own utility.
The resulting TNG bound matches that in Theorem 4.2 of \citep{donmez2024team}.
In contrast, our analysis explicitly accounts for belief-estimation errors induced by supervisor-network learning.
Accordingly, DTOA replaces the common-belief requirement in team-FP with supervisor-based distributed belief learning.
DTOA enables agents to optimize their strategies for team orchestration while ensuring the convergence of belief-estimation errors and providing a small-TNG guarantee. 
These results provide an affirmative answer to Problem \ref{problem_1}.

%% file: sections_IEEE/6_Robustness.tex
\section{Byzantine resilience}
\label{sec:algorithm_2}
% \vspace{-3pt}

This section addresses Problem \ref{problem_2} by developing a supervisor-based Byzantine-identification mechanism and adapting DTOA to the Byzantine attack setting.
We consider a Byzantine attack setting for ZSPTGs with a supervisor network, adapted from Byzantine-resilient multi-agent learning \citep{li2022byzantine} and imperfect verification \citep{haeberlen2006case}.
In this setting, each team is either honest or Byzantine; that is, $\mathcal{H} \cap \mathcal{B}=\emptyset$ and $\mathcal{H} \cup \mathcal{B}=\mathcal{T}$.
Agents in honest teams update their actions according to \eqref{agent_action_update} and report their joint actions truthfully to their supervisors, whereas agents in Byzantine teams update their actions arbitrarily and can misreport their joint actions.
Supervisors aim to identify Byzantine teams and mitigate their influence on honest teams.
The supervisors' checking policy and the Byzantine teams' attack policy are specified below.

% \vspace{-3pt}
\textbf{Supervisors' checking policy.}
For a supervisor $s\in\mathcal{S}$ and a team $m \in \mathcal{T}^s$, supervisor $s$ checks, with a prescribed probability, whether the joint action reported by team $m$ is consistent with the joint action actually taken by team $m$.
We use a binary variable $v_k^{m,s}\in\{0,1\}$ to indicate this decision, where $v_k^{m,s}=1$ means that supervisor $s$ checks team $m$ at round $k$.
After checking team $m$ at round $k$, supervisor $s$ receives a verification signal $z_k^{m,s} \in \{\mathrm{fail},\mathrm{pass}\}$.
The verification signals $z_k^{m,s}$ are not necessarily accurate.
Specifically, we assume that there exist constants $\eta_{FP} \in (0,1]$ and $\eta_{FN} \in [0,1)$ such that, for any $s \in \mathcal{S}$, $m \in \mathcal{T}^s$, $k \ge 0$,
% \vspace{-3em}
\begin{align*}
    &\mathbb{P}\left( z_k^{m,s}=\mathrm{fail} \mid v_k^{m,s}=1,\ \underline{a}_{k,r}^m = \underline{a}_k^m \right) \le \eta_{FP},\\
    &\mathbb{P}\left( z_k^{m,s}=\mathrm{pass} \mid v_k^{m,s}=1,\ \underline{a}_{k,r}^m \neq \underline{a}_k^m \right) \le \eta_{FN}.
\end{align*}
% \vspace{-2em}
Here, $\eta_{FP}$ and $\eta_{FN}$ denote upper bounds on the false-positive and false-negative probabilities, respectively.
% Smaller values of $\eta_{FP}$ and $\eta_{FN}$ correspond to more accurate verification.

% \begin{assumption}
%     There exist small $\eta_{FP} \in (0,1]$ and $\eta_{FN} \in [0,1)$ such that for any $m \in \mathcal{T}$, $s \in \mathcal{S}$ and $k \ge 0$
%     \begin{align*}
%         &\mathbb{P}\left( z_k^{s,m}={fail} \big\vert v_k^{s,m}=1,\ \underline{a}_{k,r}^m = \underline{a}_k^m \right) \le \eta_{FP},\\
%         &\mathbb{P} \left( z_k^{s,m}={pass} \big\vert v_k^{s,m}=1,\ \underline{a}_{k,r}^m \neq \underline{a}_k^m \right) \le \eta_{FN}.
%     \end{align*}
%     $\eta_{FP}$ and $\eta_{FN}$ represents the probabilities of getting false positive and false negative verification signals.
%     Note that smaller $\eta_{FP}$ and $\eta_{FN}$ imply higher verification signal accuracy.
% \end{assumption}
% \vspace{-3pt}
\textbf{Byzantine attack policy.}
A Byzantine team can be identified more readily if it misreports its joint action at each round.
We therefore assume that there exists $p_{\mathrm{lie}} \in [0,1]$ such that each Byzantine team misreports its joint action to its supervisors with probability $p_{\mathrm{lie}}$ at each round.
If a Byzantine team does not misreport at a given round, then it reports its true joint action at that round.
In this case, the Byzantine team is indistinguishable from an honest team from the supervisors' perspective, because the verification signal depends only on whether the reported joint action is consistent with the joint action actually taken. Moreover, supervisors have no access to any agent's utility function or any team's potential function.

% \vspace{-3pt}
In practice, a verification signal is meaningful only if it is informative about whether a team misreports.
We impose the following consistency condition on the verification signals.
Specifically, $\eta_{FP}$, $\eta_{FN}$, and $p_{\mathrm{lie}}$ satisfy
% \vspace{-3em}
\begin{align}
\label{frequency}
        \eta_{FP} < (1-p_{\mathrm{lie}})\eta_{FP}+p_{\mathrm{lie}}(1-\eta_{FN}) .
\end{align}
% \vspace{-2em}
Condition \eqref{frequency} means that a checked Byzantine team has a higher probability of generating a $\mathrm{fail}$ signal than a checked honest team.
Equivalently, supervisors are more likely to receive a $\mathrm{pass}$ signal after checking an honest team than after checking a Byzantine team.

% \vspace{-3pt}
The Byzantine attack setting described above is referred to as a Byzantine ZSPTG (B-ZSPTG) with a supervisor network and is defined by the tuple 
% \vspace{-3em}
\begin{align*}
    \mathcal{G}_{B}=(\mathcal{I},\mathcal{H},\mathcal{B},\mathcal{S},ST,\{SN_k\}_{k\ge 0},\{ \mathcal{A}^i \}_{i \in \mathcal{I}},\{ u^i \}_{i \in \mathcal{I}},p),
\end{align*}
% \vspace{-2em}
where $\mathcal{H}$ and $\mathcal{B}$ denote the sets of honest and Byzantine teams, respectively, and $p=(\eta_{FP},\eta_{FN},p_{\mathrm{lie}})$ denotes the parameters satisfying \eqref{frequency}.
This formulation serves as the basis for the subsequent analysis of Problem \ref{problem_2}.

\subsection{Algorithm Design}
% \vspace{-3pt}

\begin{algorithm}[t]
    \caption{Byzantine-resilient DTOA (BR-DTOA)}
    \label{alg:robust_team_fp_with_decentralized_belief_learning}
    \begin{algorithmic}[1] 
        \STATE  \textbf{Initialise:} $\{ \pi_{s,0}^m \}_{m \in \mathcal{T}}$ and $\{a_{-1}^i \}_{i \in \mathcal{I}}$ arbitrarily, $F_0^{m,s}=0$
        % , $\{ \lambda^m_0=\lambda_{min} \}_{m \in \mathcal{T}}$

        \WHILE{round $k = 0, 1, \dots$}
            \FOR{$m \in \mathcal{T}$}
                \STATE select agent $i \in \mathcal{I}^m$ in team $m$ randomly
                \IF{$m$ is an honest team}
                    \STATE select $s \in \left\{s : i \in \mathcal{T}^m \right\}$ randomly to provide belief estimates about other teams
                    \STATE agent $i$ updates its action based on $a_{k-1}^{-i}$ and $\pi_k^{m,s}$: $a_k^i \sim br_{\tau} \left( u^i(\cdot,a_{k-1}^{-i},\pi_k^{m,s}) \right)$
                \ELSE
                    \STATE agent $i$ updates its action randomly
                \ENDIF

                \FOR {$j \in \mathcal{I}^m \backslash \{ i \}$}
                    \STATE repeat the last action: $a_k^j=a_{k-1}^j$
                    % \STATE $\lambda_{k+1}^m=\lambda_{k}^m$
                \ENDFOR
                \STATE report the joint action $\underline{a}_{k,r}^m$ to supervisors $\mathcal{T}^m$
                % \STATE update $\lambda^m_{k+1}=\mathop{max}\left\{\lambda^m_k+\frac{1}{t_{\lambda,k}} c^m\left (\underline{a}_{k}^m \right),\lambda_{min} \right\}$
            \ENDFOR
            \FOR{$s \in \mathcal{S}$ and $m \in \mathcal{T}$}
                
                \IF{$m \in \mathcal{T}^s$}
                    \STATE check team $m$ for misreporting with probability $q$ and receive $z_k^{m,s} \in \{ \mathrm{pass},\mathrm{fail} \}$
                    \vspace{0.2em}
                    \STATE update $F_k^{m,s}=F_{k-1}^{m,s}+{1}\left(z_k^{m,s}=fail\right)$
                    \vspace{0.3em}
                    \STATE update $f_k^{m,s}={F_k^{m,s}}/{k}$
                    \vspace{0.3em}
                    \IF {$k<K$ or $f_k^{m,s}<f$}
                    \STATE update: $\pi_{k+1}^{m,s}\!=\!\pi_{k}^{m,s}\!+\alpha_k \!\left( \!\underline{a}_{k,r}^m \!-\!\pi_{k}^{m,s} \!\right)$
                    \ENDIF
                    % \STATE $\pi_{k+1}^{m,s}=\pi_{k}^{m,s}+\alpha_k \left( \underline{a}_{k,r}^m -\pi_{k}^{m,s} \right)$
                \ELSE
                    \STATE update: $\pi_{k+1}^{m,s}\!\!=\!\!
                    \begin{cases}
                        \!\frac{1}{\vert \mathcal{N}^s_k \vert}\!\!\sum_{s' \!\in \mathcal{N}^s_k} \!\!\pi_{k}^{m,s'}\!\!\!, \!\!&\text{for }\vert \mathcal{N}^s_k \vert\!\!>\!\!0\\
                        \pi_{k}^{m,s}, \!\!&\text{for }\vert \mathcal{N}^s_k \vert\!\!=\!\!0
                    \end{cases}$
                \ENDIF
            \ENDFOR            
        \ENDWHILE
    \end{algorithmic}
\end{algorithm}

To address the Byzantine attack setting described above, we extend DTOA with a Byzantine-resilient mechanism that operates at the supervisor-team level.
The resulting algorithm, referred to as BR-DTOA, is summarized in Algorithm \ref{alg:robust_team_fp_with_decentralized_belief_learning}.

% \vspace{-3pt}
\textbf{Byzantine-resilient mechanism.}
Since each checking outcome can be inaccurate, supervisors do not rely on a single $\mathrm{fail}$ verification signal to identify Byzantine teams.
Instead, they accumulate the evidence over time.
For a supervisor $s\!\in\!\mathcal{S}$ and a team $m \!\in\! \mathcal{T}^s$, let $F_k^{m,s}$ denote the number of times that supervisor $s$ has received a $\mathrm{fail}$ signal from team $m$ up to round $k$.
Supervisor $s$ checks team $m$ with probability $q$ and updates $F_k^{m,s}$ as
% \vspace{-3em}
\begin{align*}
    F_k^{m,s}=F_{k-1}^{m,s}+\mathbf{1}\{v_k^{m,s}=1,\ z_k^{m,s}=\mathrm{fail}\}.
\end{align*}
% \vspace{-2em}
The empirical frequency of receiving a $\mathrm{fail}$ signal from team $m$ is then defined as $f_k^{m,s}\!=\!F_k^{m,s}/k$.
To reduce the effect of early-stage fluctuations, supervisors start labeling teams only after a burn-in period of $K$ rounds.
Given a threshold $f$ satisfying $q \eta_{FP} \!<\! f\! <\! q \left[(1\!-\!p_{\mathrm{lie}})\eta_{FP}\!+\!p_{\mathrm{lie}}(1\!-\!\eta_{FN}) \right]$, if $k>K$ and $f_k^{m,s}>f$ hold, then supervisor $s$ labels team $m$ as Byzantine.
Such an $f$ exists by Condition \eqref{frequency}.

% The threshold $f$ for the empirical $\mathrm{fail}$ frequency is chosen satisfying $q \eta_{FP} < f < q  \left[(1-p_{lie})\eta_{FP}+p_{lie}(1-\eta_{FN}) \right]$.
% \begin{align*}
%     q \eta_{FP} < f < q  \left[(1-p_{lie})\eta_{FP}+p_{lie}(1-\eta_{FN}) \right].
% \end{align*}
% \vspace{-5pt}
Once a team $m$ is labeled as Byzantine by a supervisor $s$, supervisor $s$ keeps this label in all subsequent rounds and stops providing belief estimates to team $m$.
This label is also used to prevent identified Byzantine teams from further affecting supervisors' belief estimates.
At round $k$, supervisor $s$ updates its belief regarding team $m$ according to
% \vspace{-3em}
\begin{align*}
    \pi_{k\!+\!1}^{m,s}\!\!\!=\!\!
    \begin{cases}
        \!\pi_{k}^{m,s}\!\!\!+\!\omega_k^{m,s} \!\alpha_k\!\! \left(\! \underline{a}_{k,r}^m \!\!-\!\pi_{k}^{m,s} \!\right)\!, &\!\!\!\! \text{for } m \!\in \!\mathcal{T}^s,\\
       \!\frac{1}{\vert \mathcal{N}^s_k \vert}\sum_{s' \in \mathcal{N}^s_k} \pi_{k}^{m,s'}, & \!\!\!\!\text{for } m \!\!\notin\!\! \mathcal{T}^s \text{and } \vert \mathcal{N}^s_k \vert\!\!>\!\!0,\\
       \!\pi_{k}^{m,s}, & \!\!\!\!\text{for } m \!\!\notin\!\! \mathcal{T}^s \text{and } \vert \mathcal{N}^s_k \vert\!\!=\!\!0,
    \end{cases}   
\end{align*}
% \vspace{-2em}
where $\omega_k^{m,s}=1$ if team $m$ has not been labeled as Byzantine by supervisor $s$, and $\omega_k^{m,s}=0$ otherwise.

\subsection{Resilience Analysis}
% \vspace{-3pt}
We analyze BR-DTOA in two steps.
We first examine the Byzantine-resilient mechanism and then study the convergence of BR-DTOA and the associated honest TNG.
The following lemma provides exponential bounds on Byzantine-team identification errors.

\begin{lemma}
\label{theorem_of_checking_performance}
Given a B-ZSPTG with a supervisor network $\mathcal{G}_{B}=(\mathcal{I},\mathcal{H},\mathcal{B},\mathcal{S},ST,\{SN_k\}_{k\ge 0},\{ \mathcal{A}^i \}_{i \in \mathcal{I}},\{ u^i \}_{i \in \mathcal{I}},p)$, for any supervisor $s$ and any team $m \in \mathcal{T}^s$, the following statements hold under BR-DTOA:

% \vspace{-1em}
\begin{itemize}
    \item if $m \in \mathcal{H}$, then $\mathbb{P}\left( f_k^{m,s}\! >\!f \right) \!\le\! \mathop{\exp} \left( -kD(f \Vert q \eta_{FP}) \right)$;
    \item if $m \in \mathcal{B}$, then
        % \vspace{-1em}
        \begin{align*}
            \mathbb{P}\left( f_k^{m,s} \!\!\!<\!\!f \right) 
            \!\!\le\!\! \mathop{\exp} \!\Big( \!\!\!-\!\!kD\big(f \Vert q \!\left[(\!1\!\!-\!\!p_{\mathrm{lie}}\!)\eta_{F\!P}\!+\!p_{\mathrm{lie}}(\!1\!\!-\!\!\eta_{F\!N}\!) \right]\! \big) \!\Big);
        \end{align*}
\end{itemize}

% \vspace{-1em}
where $D(x \Vert y) \triangleq x \log  \left(\frac{x}{y}\right)+(1-x)\log \left(\frac{1-x}{1-y}\right)$ denotes the binary relative entropy.
\end{lemma}

% \vspace{-3pt}
Lemma \ref{theorem_of_checking_performance} shows that the probabilities of falsely labeling an honest team as Byzantine and failing to identify a Byzantine team by round $k$ decay exponentially as $k \rightarrow \infty$.
A larger checking probability $q$ and smaller error rates $\eta_{FP}$ and $\eta_{FN}$ can improve the identification performance by reducing the misidentification probabilities.
The threshold $f$ controls the trade-off between missed identifications of Byzantine teams and false alarms for honest teams.
Increasing $f$ reduces the probability of falsely labeling an honest team as Byzantine but increases the probability of failing to identify a Byzantine team.
Conversely, decreasing $f$ facilitates Byzantine-team identification but raises the risk of misclassifying honest teams.

% \vspace{-3pt}
We now derive convergence guarantees for BR-DTOA from the identification result in Lemma \ref{theorem_of_checking_performance} and the convergence analysis of DTOA in the previous section.
The first result concerns supervisors' belief-estimation errors for honest teams.

\begin{theorem}
\label{byzantine_belief_errors_convergence}
For a Byzantine ZSPTG with a supervisor network $\mathcal{G}_{\!B}\!\!=\!\!\left(\mathcal{I},\!\mathcal{H},\!\mathcal{B},\!\mathcal{S},\!ST,\!\{\!SN_k\!\}_{k\ge 0},\!\{ \!\mathcal{A}^i \!\}_{i \in \mathcal{I}},\!\{ \!u^i\! \}_{i \in \mathcal{I}},\!p\right)$,
% \begin{align*}
%     \mathcal{G}_{B}=(\mathcal{I},\mathcal{H},\mathcal{B},\mathcal{S},ST,\{SN_k\}_{k\ge 0},\{ \mathcal{A}^i \}_{i \in \mathcal{I}},\{ u^i \}_{i \in \mathcal{I}},p),
% \end{align*}
under the conditions of Theorem \ref{belief_convergence_dynamic}, for any supervisor $s \in \mathcal{S}$, the following statement holds for BR-DTOA with probability at least $1-\delta_B$:
% \vspace{-3em}
\begin{align}
\label{theorem_3_statement}
    \Vert \pi_k^{m,s}-\pi_k^m \Vert_{\infty}\le O(\rho^{[\frac{k}{2}]})+O(\alpha_{[\frac{k}{2}]}),\quad \forall  m \in \mathcal{H},
\end{align}
% \vspace{-2em}
where $\delta_B=\delta_B^{\mathcal{H}}+\delta_B^{\mathcal{B}}$ with $\delta_B^{\mathcal{H}}=\vert \mathcal{H} \vert \cdot \frac{\mathop{\exp} \left( -K \cdot D(f \Vert q \eta_{FP}) \right)}{1-\mathop{\exp} \left( -D(f \Vert q \eta_{FP}) \right)}$ and $\delta_B^{\mathcal{B}}
=\vert \mathcal{B} \vert \cdot \mathop{\exp} \Big( -K \cdot D \big(f \Vert q \big[(1-p_{\mathrm{lie}})\eta_{FP}+p_{\mathrm{lie}}(1-\eta_{FN}) \big] \big) \Big)$, and $\rho$ is defined as in Theorem \ref{belief_convergence_dynamic}.
If Assumption \ref{step_size_assumption} also holds, (\ref{theorem_3_statement}) simplifies to $\Vert \pi_k^{m,s}-\pi_k^m \Vert_{\infty}\le O(\alpha_{[\frac{k}{2}]})$.
\end{theorem}

% \vspace{-3pt}
Lemma \ref{theorem_of_checking_performance} shows that, after the burn-in period, BR-DTOA avoids falsely labeling honest teams and identifies Byzantine teams with high probability.
On this event, Theorem \ref{belief_convergence_dynamic} can be applied to bound the belief-estimation errors for honest teams, yielding Theorem \ref{byzantine_belief_errors_convergence}.
Thus, the convergence rate remains the same as in the non-Byzantine case, while the guarantee holds with probability at least $1-\delta_B$ because of possible identification errors.
Moreover, $\delta_B \rightarrow 0$ as $K \rightarrow \infty$, so a longer burn-in period improves the reliability.
The next result establishes the honest-TNG convergence guarantee under BR-DTOA.

\begin{figure*}[t]
    \centering

    \subfloat[Three supervisors.]{
        \begin{minipage}[t]{0.27\textwidth}
            \centering
            \includegraphics[width=\linewidth,trim=0 0 20 33, clip]{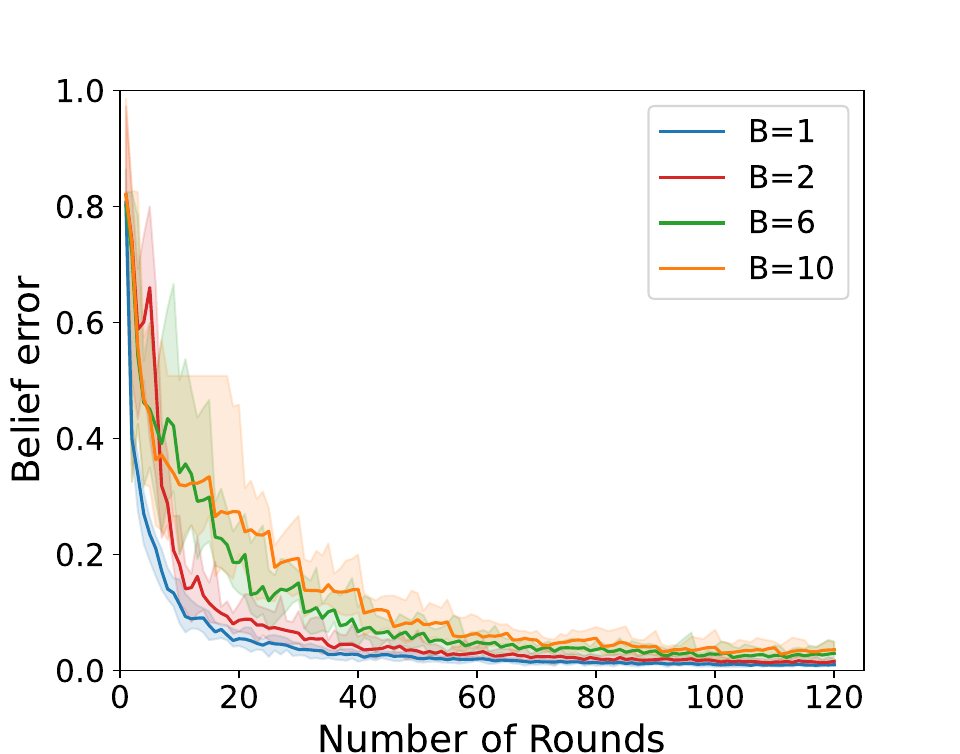}
            \label{experiment_1_1}
        \end{minipage}
    }
    \hfill
    \subfloat[Five supervisors.]{
        \begin{minipage}[t]{0.27\textwidth}
            \centering
            \includegraphics[width=\linewidth,trim=0 0 20 33, clip]{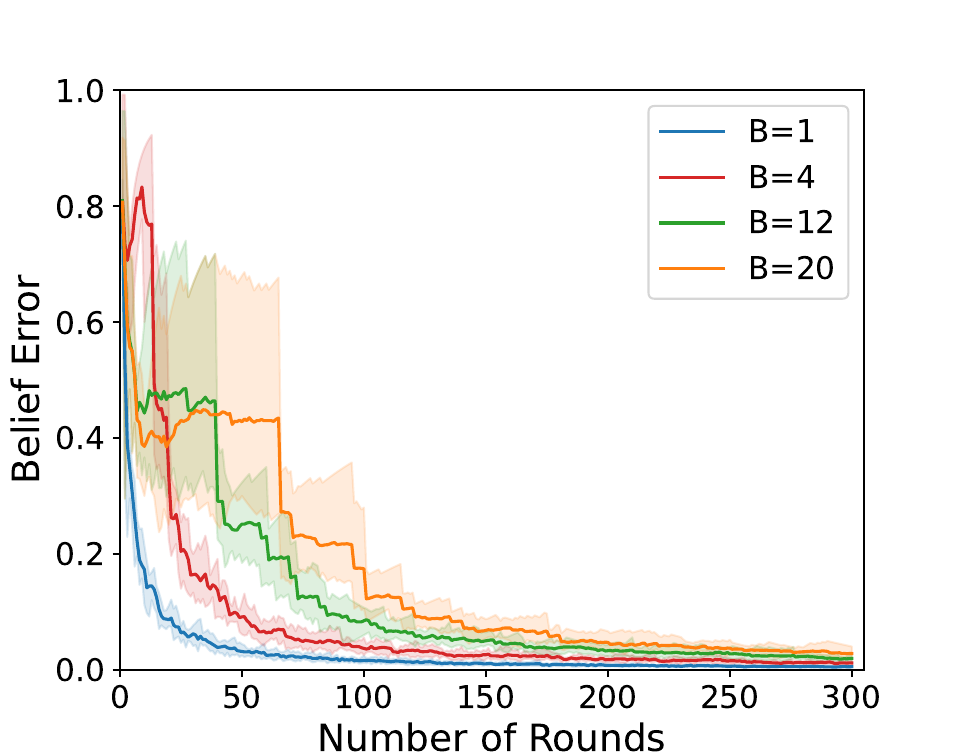}
            \label{experiment_1_2}
        \end{minipage}
    }
    \hfill
    \subfloat[Rounds needed for small errors.]{
        \raisebox{0.05em}{
            \begin{minipage}[t]{0.4\textwidth}
                \centering
                \includegraphics[width=\linewidth,trim=0 10 0 0, clip]{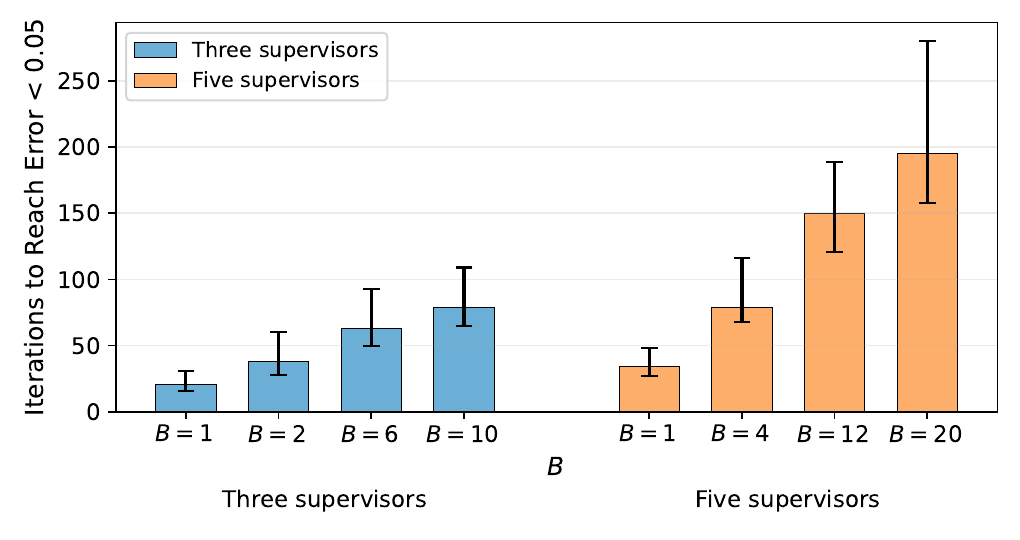}
                \label{experiment_1_3}
            \end{minipage}
        }
    }

    \caption{Convergence of belief-estimation errors.}
    \label{fig:belief_error}
\end{figure*}

\begin{theorem}
\label{honest_TNG}
For a Byzantine ZSPTG with a supervisor network $\mathcal{G}_{\!B}\!\!=\!\!\left(\mathcal{I},\!\mathcal{H},\!\mathcal{B},\!\mathcal{S},\!ST,\!\{\!SN_k\!\}_{k\ge 0},\!\{ \!\mathcal{A}^i \!\}_{i \in \mathcal{I}},\!\{ \!u^i\! \}_{i \in \mathcal{I}},\!p\right)$,
under the conditions of Theorem \ref{convergence_of_gap}, for any supervisor $s \in \mathcal{S}$, the following inequality holds for BR-DTOA with probability at least $1-\delta_B$
% \vspace{-3em}
\begin{align*}
     &\limsup_{k \to \infty}TNG_{\mathcal{H}}\left(\pi_k^s\right)\nonumber\\
     \le&
     \begin{cases}
        \tau \log \left\vert \mathcal{A}_{\mathcal{H}} \right\vert, &\text{for BR-DTOA},\\
        \tau \log \left\vert \mathcal{A}_{\mathcal{H}} \right\vert+\left\vert \mathcal{H}  \right\vert^2 \overline{\phi} \Lambda(\delta,\epsilon), &\text{for BR-iDTOA},
    \end{cases}
\end{align*}
% \vspace{-2em}
where $\delta_B$ is defined as in Theorem \ref{byzantine_belief_errors_convergence}, and $\mathcal{A}_{\mathcal{H}}=\prod_{m \in \mathcal{H}} \underline{\mathcal{A}}^m$ is the joint action set of honest teams.
\end{theorem}

% Theorem \ref{honest_TNG} provides a probabilistic performance guarantee for Algorithm \ref{alg:robust_team_fp_with_decentralized_belief_learning}.
% The probability parameter $\delta_B$ quantifies the risk caused by imperfect identification, including false alarms for honest teams and missed identifications of Byzantine teams.
% Compared with Theorem \ref{convergence_of_gap}, the honest team-Nash gap admits the same type of temperature-dependent bound after the identified Byzantine teams are excluded.
% As $\delta_B \rightarrow 0$, the resulting upper bound on the honest TNG approaches the Byzantine-free TNG bound in Theorem \ref{convergence_of_gap} restricted to honest teams.
% The above results show that supervisors can identify Byzantine teams with vanishing error probabilities and that the honest TNG remains asymptotically bounded after excluding the identified teams. 
% Thus, the proposed Byzantine-resilient TOA provides an affirmative answer to Problem \ref{problem_2}.

% Theorem \ref{honest_TNG} extends the TNG convergence result of Theorem \ref{convergence_of_gap} to the Byzantine attacking setting. 
% Conditional on the high-probability event characterised by Lemma \ref{theorem_of_checking_performance}, the honest-team learning dynamics satisfy the same type of TNG bound as in the non-Byzantine case, restricted to honest teams. 
% This restriction uses the identified Byzantine-team information and therefore yields a sharper guarantee than treating all teams uniformly.

% \vspace{-3pt}
Theorem \ref{honest_TNG} extends the TNG result in Theorem \ref{convergence_of_gap} to the Byzantine attack setting.
Conditional on the high-probability event characterized by Lemma \ref{theorem_of_checking_performance}, the honest-team learning dynamics satisfy an honest-TNG bound of the same form as in the non-Byzantine case, but with the gap evaluated only over honest teams. This yields a sharper guarantee than treating all teams uniformly.

% \vspace{-3pt}
Together, BR-DTOA can identify Byzantine teams with vanishing error probability and the Byzantine resilience of BR-DTOA is further established by showing that, in the presence of Byzantine teams, it preserves the convergence of belief-estimation errors for honest teams and provides an honest-TNG guarantee.
These results affirmatively answer Problem \ref{problem_2}.

%% file: sections_IEEE/7_Experiments.tex
\section{Experimental Validation}
\label{sec:experiments}

% \vspace{-3pt}
In this section, we present numerical experiments to illustrate the theoretical results in Sections \ref{sec:algorithm_1} and \ref{sec:algorithm_2}, including the convergence, optimality, and Byzantine resilience of the proposed algorithms.
We also examine the scalability of DTOA.
We further compare DTOA with two baseline methods and numerically extend DTOA to the Markov decision process (MDP) setting.

\begin{figure}[t]
    \centering
%    \vspace{0em}
    \subfloat[DTOA.]{
        
        \includegraphics[width=0.47\linewidth,trim=0 5 10 10, clip]{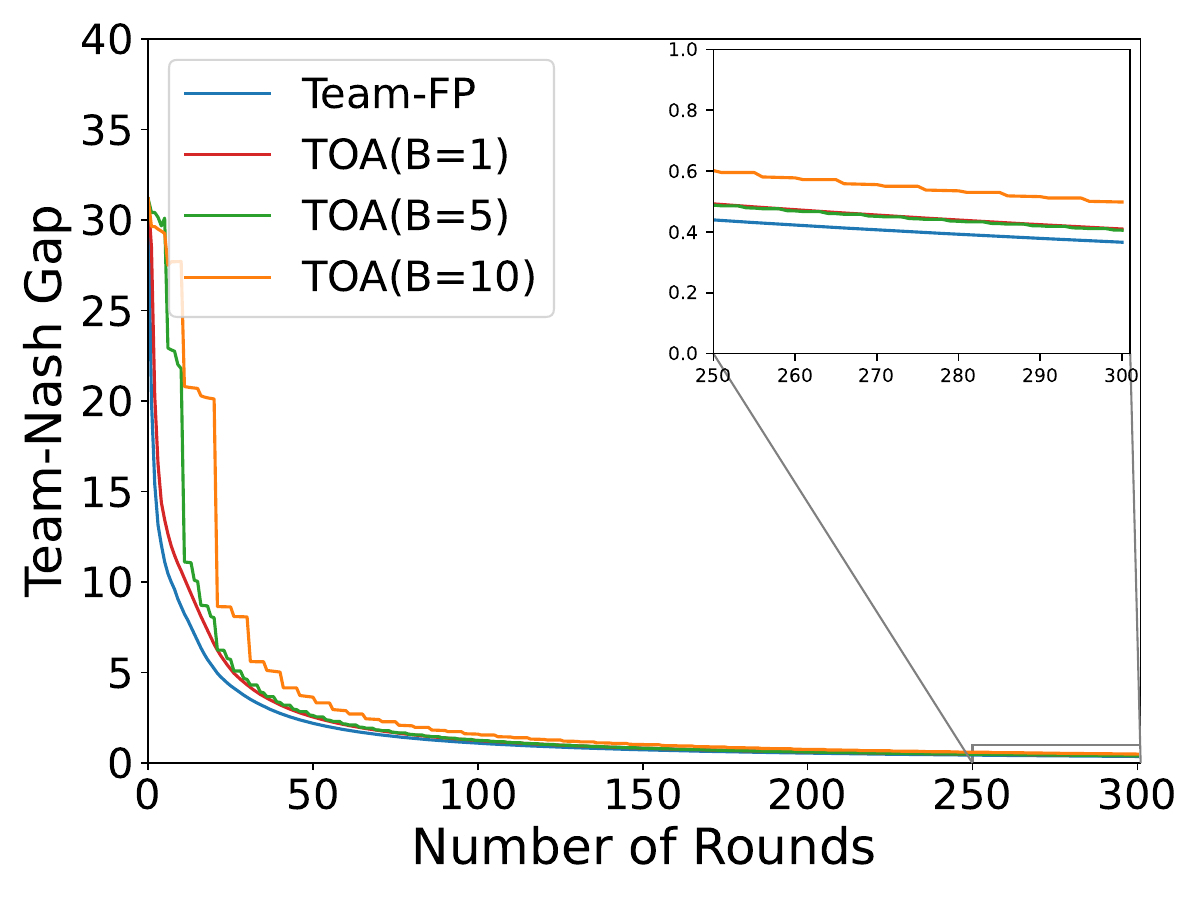}
        \label{experiment_2_1}
    }
    \subfloat[Independent DTOA.]{
        \includegraphics[width=0.47\linewidth,trim=0 5 10 10, clip]{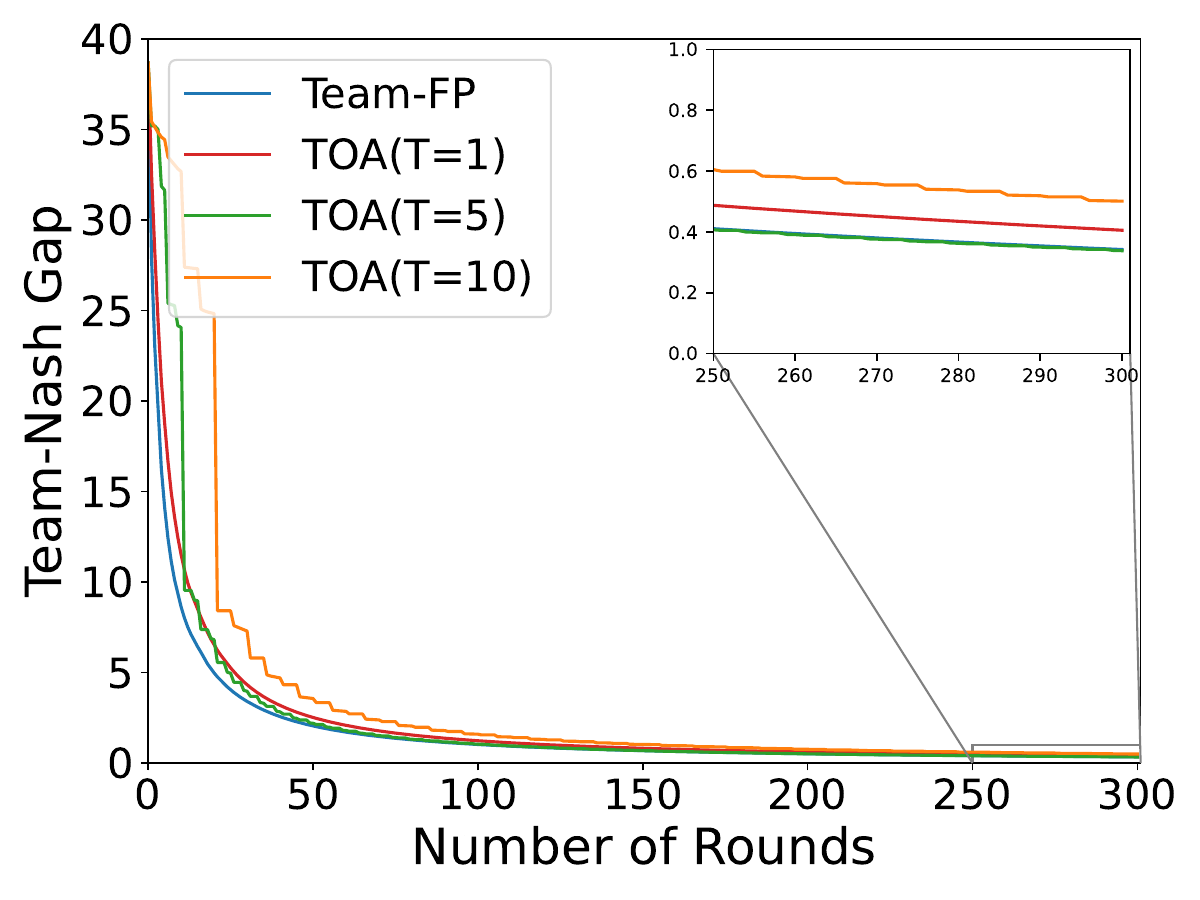}
        \label{experiment_2_2}
    }
    \caption{Convergence of TNG. }
    \label{fig:TNG}
\end{figure}

% \vspace{-8pt}

\begin{figure}[t]
\vspace{-0.1em}
    \centering
    \subfloat[More teams.]{
        \includegraphics[width=0.47\linewidth,trim=0 0 20 33, clip]{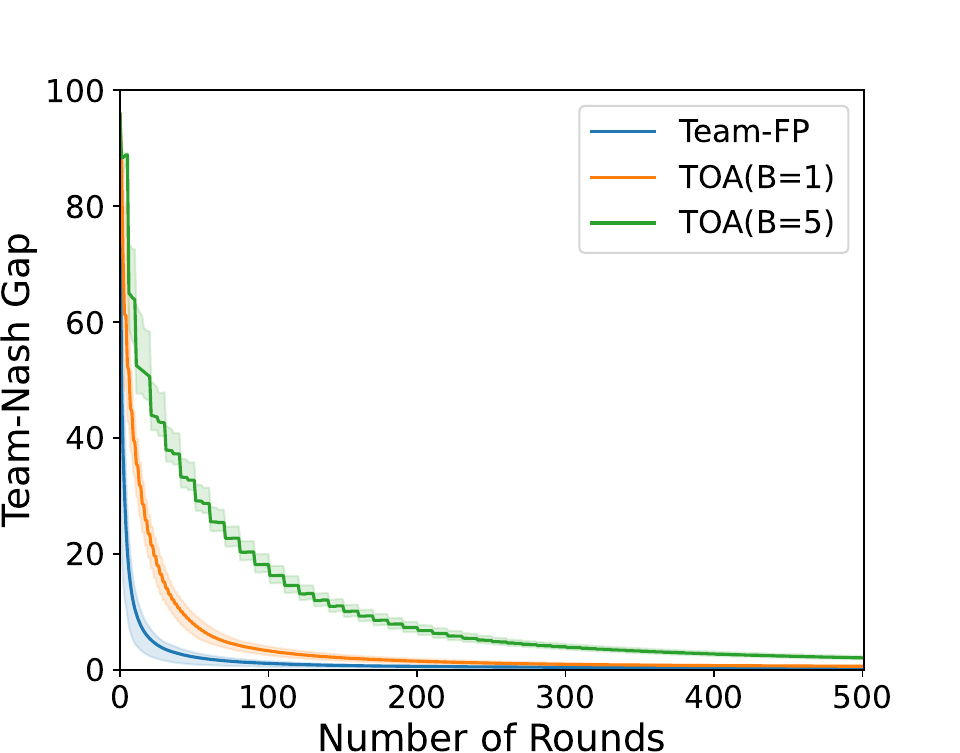}
        \label{experiment_4_1}
    }
    \subfloat[More agents per team.]{
        \includegraphics[width=0.47\linewidth,trim=0 0 20 33, clip]{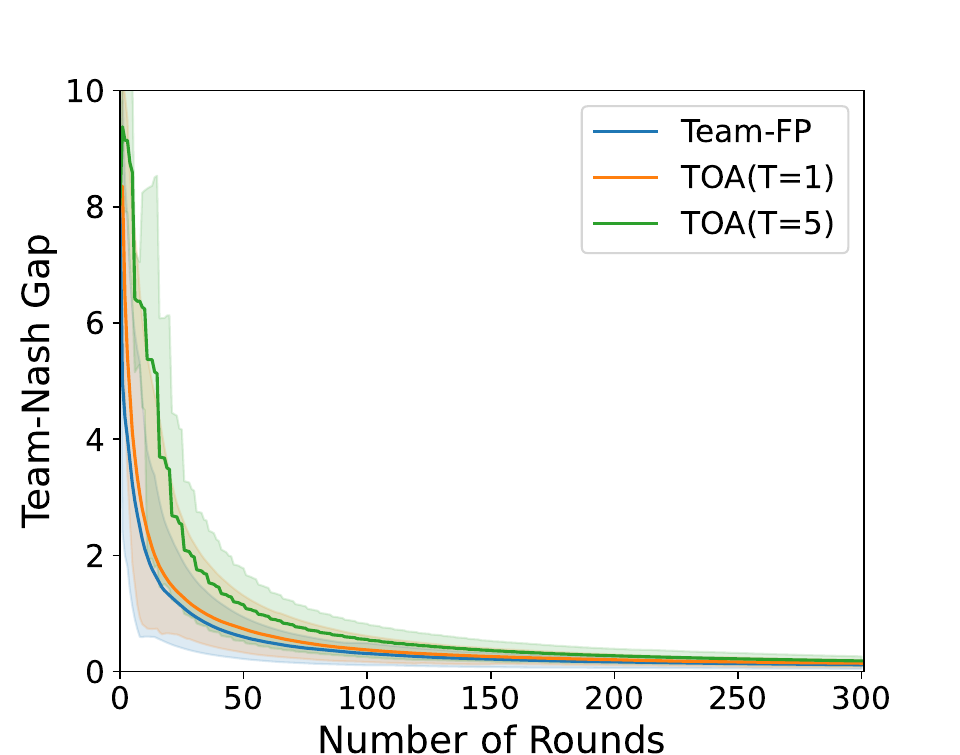}
        \label{experiment_4_2}
    }
    \caption{Scalability tests.}
    \label{fig:scalability}
\end{figure}

% \vspace{-3pt}
Unless otherwise specified, the experiments are conducted on a repeated ZSPTG with a supervisor network consisting of five teams, three supervisors, and two agents in each team.
Each agent has the binary action set $\{0,1\}$.
The utility and potential functions are chosen separately for different experiments to reflect the corresponding scenarios while preserving the ZSPTG structure.
The supervision relationship between teams and supervisors is given by $\mathcal T^1\!\!=\!\!\{1,2\}$, $\mathcal T^2\!\!=\!\!\{3,4\}$, and $\mathcal T^3\!\!=\!\!\{5\}$. 
We set the step size to $\alpha_k\!=\!1/(k\!+\!1)$ and the temperature parameter to $\tau\!=\!0.1$. 
The shaded regions in the figures indicate the variability across ten independent runs.
The code is available at \url{https://github.com/zjt-1229/team_game_with_supervisor_network}.

\subsection{Numerical Results for DTOA}
% \vspace{-3pt}
We first present numerical results for DTOA.
These experiments aim to illustrate the convergence behavior and TNG performance characterized by the theoretical analysis and to examine the scalability of DTOA.

\begin{figure*}[t]
    \centering

    \subfloat[Honest TNG.]{
    \begin{minipage}[t]{0.22\textwidth}
        \centering
        \raisebox{0.05em}{
            \includegraphics[width=\linewidth,trim=5 0 33 33, clip]{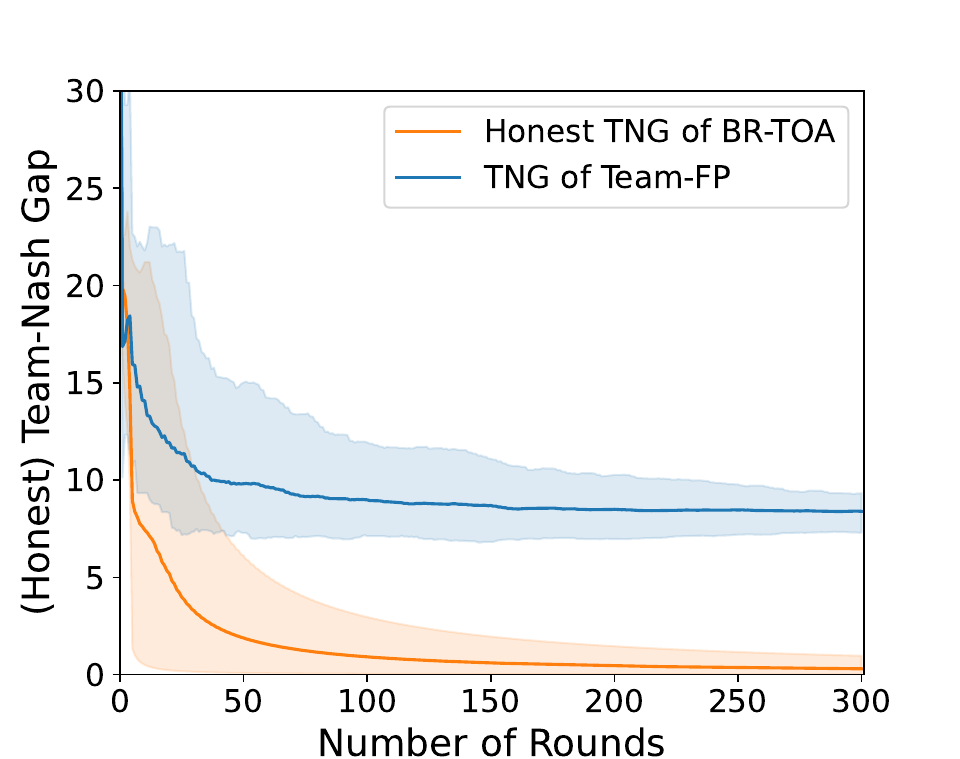}
        }
        \label{fig:experiment_3_1-1}
    \end{minipage}
    }
    \hfill
    \subfloat[Fail-signal frequencies.]{
    \begin{minipage}[t]{0.22\textwidth}
        \centering
        \raisebox{0.15em}{
            \includegraphics[width=\linewidth,trim=5 0 33 33, clip]{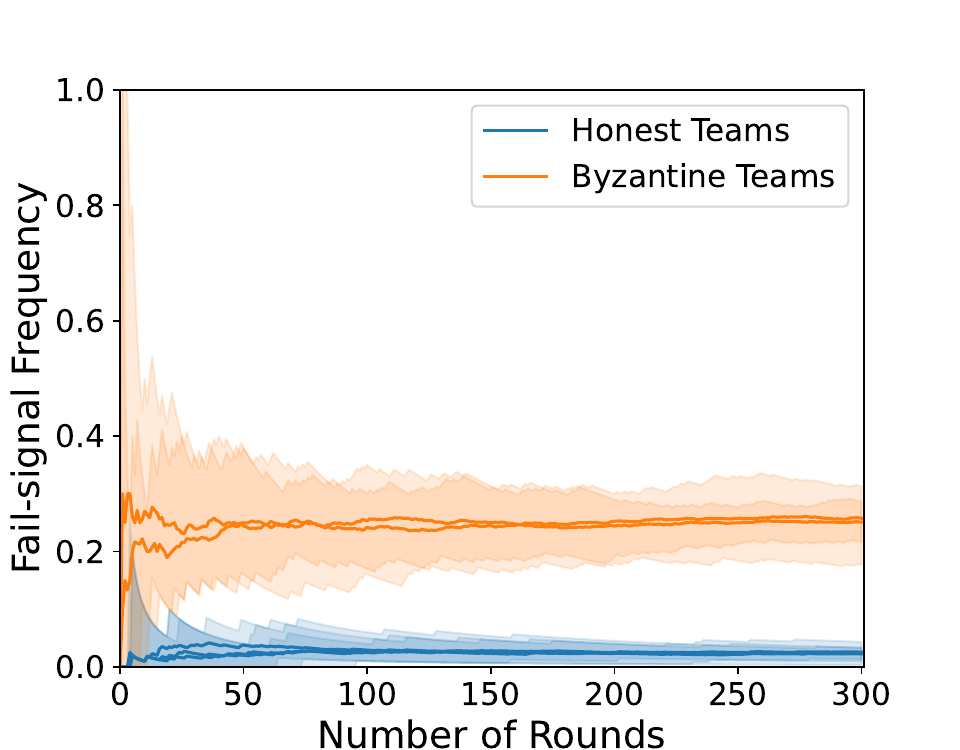}
        }
        \label{fig:experiment_3_1-2}
    \end{minipage}
    }
    \hfill
    \subfloat[Effect of $q$.]{
        \begin{minipage}[t]{0.25\textwidth}
            \centering
            \includegraphics[width=\linewidth,trim=35 65 40 40, clip]{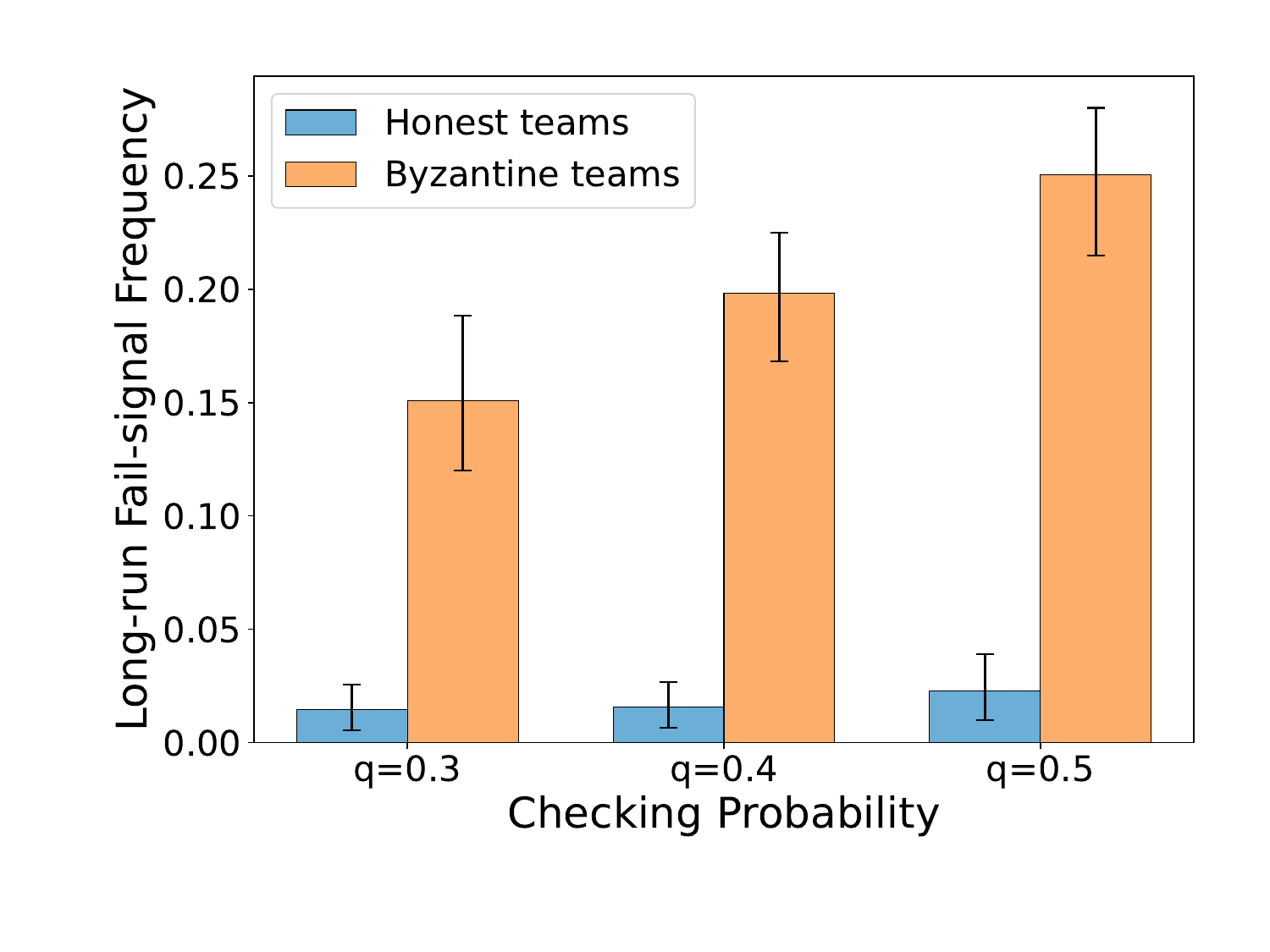}
            \label{fig:experiment_3_2}
        \end{minipage}
    }
    \hfill
    \subfloat[Effect of $p_{lie}$.]{
        \begin{minipage}[t]{0.25\textwidth}
            \centering
            \includegraphics[width=\linewidth,trim=35 65 40 40, clip]{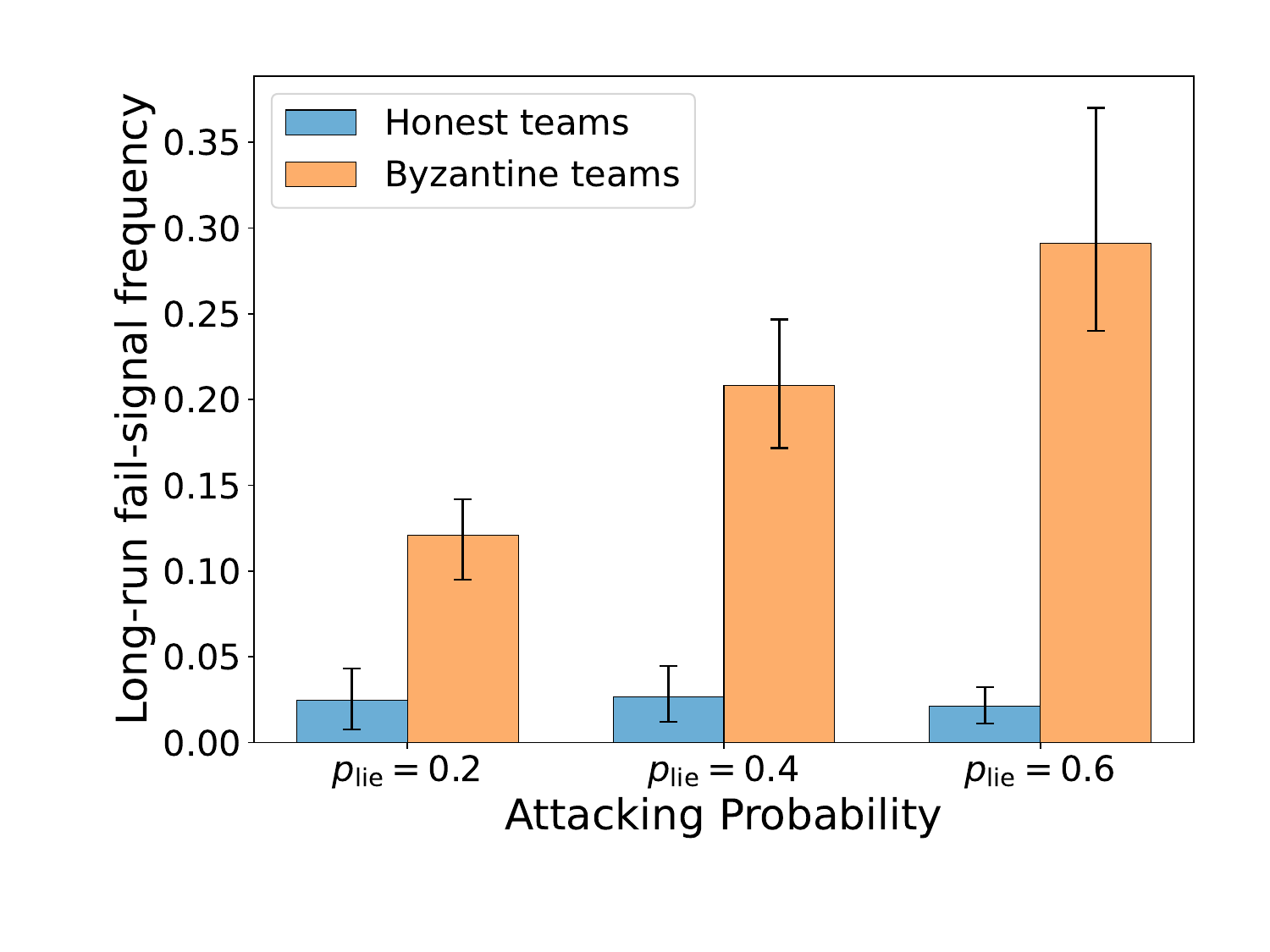}
            \label{fig:experiment_3_3}
        \end{minipage}
    }

    \caption{Performance of the BR-DTOA.}
    \label{fig:byzantine_results}
\end{figure*}

\begin{figure}[t]
\centering

\begin{minipage}[t]{0.48\linewidth}
    %\vspace{0.97em} 
    \centering
    \includegraphics[width=\linewidth,trim=0 10 5 0,clip]{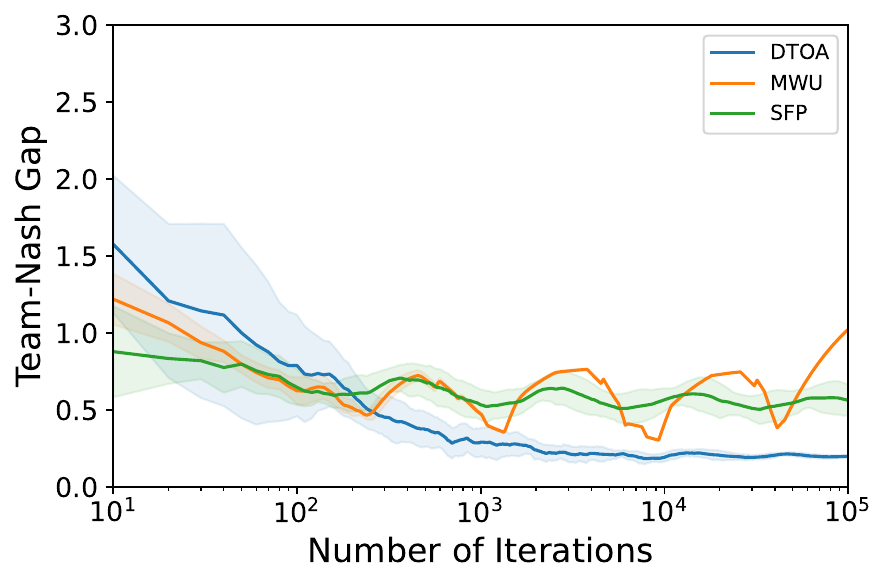}
    \captionof{figure}{Comparison.}
    \label{fig:comparison}
\end{minipage}
\hfill
\begin{minipage}[t]{0.48\linewidth}
    %\vspace{1.13em} 
    \centering
    \includegraphics[width=\linewidth,trim=0 -7 0 0,clip]{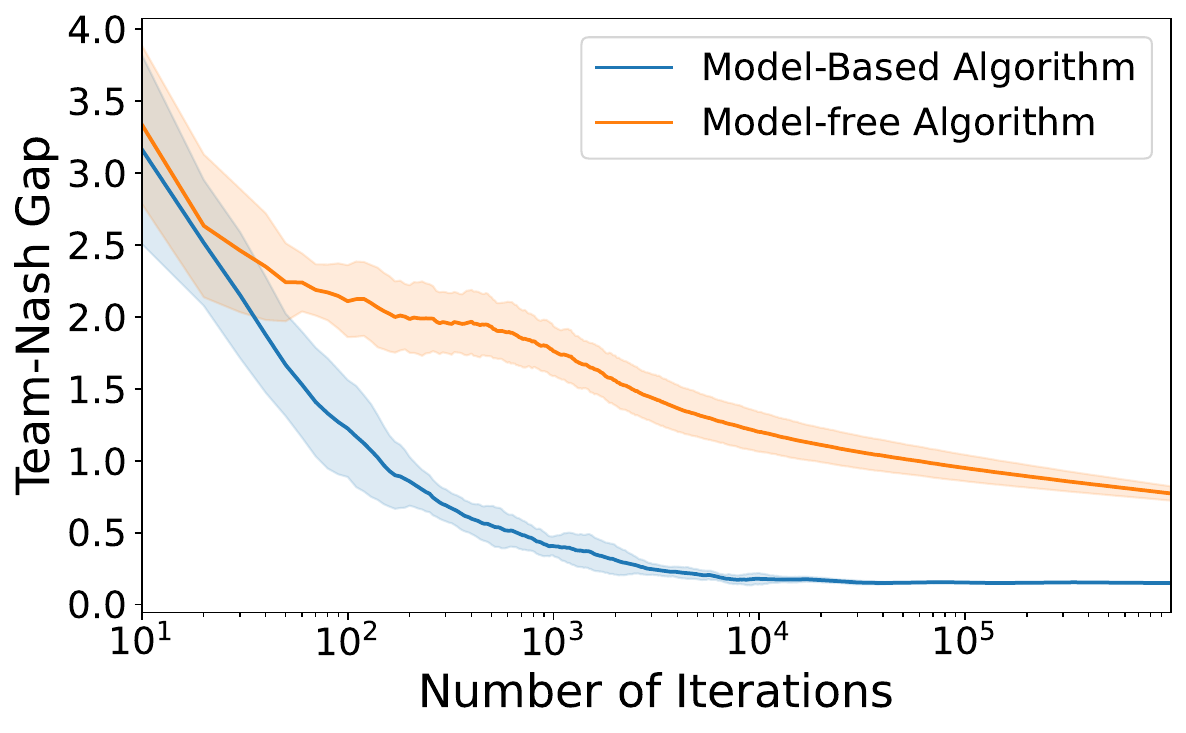}
    \captionof{figure}{MDP setting.}
    \label{fig:MDP}
\end{minipage}

\end{figure}

% \vspace{-3pt}
\textbf{Hypothesis I: convergence of belief-estimation errors.} 
Theorem \ref{belief_convergence_dynamic} shows that supervisors' belief-estimation errors converge to zero, with slower convergence when the number of supervisors $\vert \mathcal{S} \vert$ or the communication window length $B$ increases.
We examine the convergence of supervisors' belief-estimation errors in two instances: one with three supervisors and six teams, and the other with five supervisors and ten teams.
In both instances, each supervisor supervises two teams, each team contains two agents, and each agent has the binary action set $\{0,1\}$. 
Each agent randomly selects an action at every round, so the experiment isolates supervisor-network belief learning under random action sequences.
Fig. \ref{fig:belief_error} reports the results.
The errors in Fig. \ref{experiment_1_2} converge more slowly than those in Fig. \ref{experiment_1_1}, illustrating that increasing $\vert \mathcal{S} \vert$ slows the convergence rate.
Moreover, in both instances, a larger $B$ leads to slower convergence, as reflected in Fig. \ref{experiment_1_3} by the increased number of rounds required to reach a small error level.  
These observations are consistent with Theorem \ref{belief_convergence_dynamic}.

% \vspace{-3pt}
\textbf{Hypothesis II: upper bound on TNG.}
Theorem \ref{convergence_of_gap} implies that DTOA and independent DTOA converge to a near TNE in terms of the TNG, with an asymptotic bound comparable to that of team-FP \cite{donmez2024team}, while sparser supervisor communication can slow convergence.
We examine TNG convergence under DTOA and independent DTOA. 
The independent update probability in independent DTOA is set to $\delta=0.5$.
Fig. \ref{fig:TNG} shows that both DTOA and independent DTOA eventually attain TNGs comparable to those of the corresponding team-FP benchmarks, which is consistent with Theorem \ref{convergence_of_gap}.
The numerical results also show that the TNG curves under DTOA and independent DTOA converge more slowly than those under team-FP.
The slower convergence can be attributed to belief-estimation errors induced by distributed belief learning over the supervisor network.
A larger communication window length $B$ further slows convergence in both instances, suggesting that sparser supervisor communication delays the reduction of the TNG.

% \vspace{-3pt}
\textbf{Scalability tests.}
To examine the scalability, we consider two larger instances: one with increased numbers of supervisors and teams, and the other with an increased number of agents per team.
In the first instance, we consider eight supervisors and fifteen teams with two agents per team; in the second, we consider three supervisors and five teams with five agents per team. 
Fig. \ref{fig:scalability} shows that DTOA still ensures a small TNG in both larger instances, consistent with the small-scale experiments.

\subsection{Numerical Results for BR-DTOA}
% \vspace{-3pt}
We next present numerical results for BR-DTOA.
These experiments illustrate the implications of the theoretical analysis for Byzantine-team identification and honest-TNG convergence.

% We present numerical results for Byzantine-resilient DTOA to validate the empirically verifiable hypotheses on Byzantine-team identification and honest-TNG convergence suggested by the theoretical analysis.
% \vspace{-3pt}
\textbf{Hypothesis III: convergence of honest TNG.}
The theoretical results indicate that BR-DTOA identifies Byzantine teams through accumulated $\mathrm{fail}$-signal frequencies and provides an honest-TNG guarantee under Byzantine attacks.
We consider an instance with two Byzantine teams.
The parameters are set as $\eta_{FP}=\eta_{FN}=0.05$, $K=50$, $f=0.1$, $q=0.5$, and $p_{\mathrm{lie}}=0.5$.
To examine parameter sensitivity, we also vary $p_{\mathrm{lie}}$ and $q$ separately while keeping all other parameters unchanged.
Fig. \ref{fig:experiment_3_1-1} shows that the honest TNG decreases to a small value under Byzantine attacks, which is consistent with Theorem \ref{honest_TNG}.
% By contrast, team-FP does not include a Byzantine-resilient mechanism, and its TNG remains higher in this setting because Byzantine teams continue to affect the learning process.
By contrast, team-FP does not include a Byzantine-resilient mechanism and cannot distinguish Byzantine teams from honest teams. 
Consequently, the overall TNG under team-FP remains high in the Byzantine attack setting, indicating that team-FP fails to provide an effective performance guarantee for honest teams.
Fig. \ref{fig:experiment_3_1-2} shows that the $\mathrm{fail}$-signal frequencies of honest and Byzantine teams become separated after the burn-in period $K=50$.
Figs. \ref{fig:experiment_3_2} and \ref{fig:experiment_3_3} show that the $\mathrm{fail}$-signal frequencies of Byzantine teams increase as $p_{\mathrm{lie}}$ and $q$ increase, respectively. 
These observations are consistent with Lemma \ref{theorem_of_checking_performance}.

\subsection{Comparison With Baseline Methods}
% \vspace{-3pt}
We compare DTOA with multiplicative weights update (MWU) and smoothed fictitious play (SFP) on a two-team ZSPTG instance with one supervisor assigned to each team. 
For MWU and SFP, we relax the information constraint by allowing access to full opponent information, whereas DTOA relies on supervisor-based distributed belief learning. 
Fig. \ref{fig:comparison} shows the results. 
In this instance, DTOA attains a lower TNG after sufficiently many rounds, while MWU exhibits persistent oscillations and SFP stabilizes at a higher TNG level.
This suggests that DTOA can maintain a smaller TNG even under a more restrictive information structure.

% \begin{figure}[t]
%     \centering

%         \centering
%         \includegraphics[width=0.8\linewidth,trim=0 0 0 0, clip]{figures/comparison.pdf}
%         \caption{Comparison with MWU and SFP.}
%         \label{fig:comparison}

% \end{figure}

\subsection{Extension to an MDP Setting}
% \vspace{-3pt}
Although the theoretical analysis in this paper focuses on repeatedly played ZSPTGs, we also examine the proposed supervisor-network learning framework in an MDP setting.
The experiment considers three teams and three supervisors, with each team supervised by one supervisor. 
Fig. \ref{fig:MDP} shows the results. 
Both the model-based and model-free variants reduce the TNG through learning, suggesting that the proposed framework can be applied numerically in an MDP setting.
Compared with the corresponding results in \citep{donmez2024team}, the reduction in the TNG is slower, which may be attributed to the additional belief-estimation errors induced by the supervisor network.

% \begin{figure}[t]
%         \centering
%         \includegraphics[width=0.8\linewidth,trim=0 0 0 0, clip]{figures/MDP.pdf}
%         \caption{MDP setting.}
%         \label{fig:MDP}

% \end{figure}

%% file: sections_IEEE/8_Conclusion.tex
\section{Conclusions}
\label{sec:conclusion}
\vspace{-3pt}
In this paper, we studied team-orchestrating learning in repeatedly played ZSPTGs with a supervisor network under distributed belief information. 
We proposed the DTOA, which combines team-FP with distributed belief learning over a supervisor network, established the convergence of supervisors' belief-estimation errors, and showed that the induced learning dynamics converged to a near TNE with a small TNG.
We further considered a Byzantine attack setting, where Byzantine teams could misreport their joint actions and developed the BR-DTOA by integrating DTOA with a supervisor-based identification mechanism.
For BR-DTOA, we established the convergence of supervisors' belief-estimation errors for honest teams and derived an honest TNG guarantee. 
Numerical simulations illustrated the convergence of the proposed algorithms and the effectiveness of the Byzantine-resilient mechanism.

%Future work may proceed in three directions. 
%One direction is to extend the framework to mean-field team games with infinite number of agents. 
%Another direction is to study continuous or more general action spaces. 
%It is also of interest to develop theoretical guarantees for the MDP setting, which has been explored numerically in this paper.

%% file: appendices_IEEE/proof_of_belief_convergence.tex
\subsection{Proof of Theorem \ref{belief_convergence_dynamic}}
\label{proof_of_believe_convergence}
We prove Theorem \ref{belief_convergence_dynamic} by first establishing a technical result on the supervisor networks $\{SN_k\}_{k\ge0}$ and then proving the convergence of the belief-estimation errors.

\begin{lemma}
\label{dynamic_path_lemma}
Let $W_k \in \mathbb{R}^{\vert \mathcal{S} \vert \times \vert \mathcal{S} \vert}$ denote the gossip matrix associated with the supervisor network at round $k$, where 
% \vspace{-3em}
\begin{align*}
    W_k(s,s')\!=\!
\begin{cases}
   \frac{SN_k(s,s')}{\Vert SN_k(s,\cdot) \Vert_1}, &\!\!\! \text{if } \Vert SN_k(s,\cdot) \Vert_1\!>\!0,\\
   \mathbf{1} \{ s'\!=\!s \}, & \!\!\!\text{if } \Vert SN_k(s,\cdot) \Vert_1\!=\!0.
\end{cases}
\end{align*}

% \vspace{-2em}
If Assumption \ref{dynamic_assumption} holds, there exists some $L \in \mathbb{N}_+$ such that, for any supervisors $s$, $s'$ and any $k \ge 0$, the following statement holds:
%\begin{align*}
    $\big[W_{k+L} \cdots W_k\big](s,s')\ge \left[\frac{1}{\vert \mathcal{S}\vert}\right]^L$.
%\end{align*}
\end{lemma}

% \vspace{-1em}
\begin{proof}
We consider a setting in which $W_k(s,s')$ represents the probability of $s$ reaching $s'$ at round $k$; therefore, $W_k$ is a transition matrix.
A path on the dynamic supervisor network $\{SN_k\}_{k \ge 0}$ is defined as $\mathcal{P}=(X_k \in \mathcal{S})_{k \ge 0}$ such that $W_k(X_k,X_{k+1})>0$.
Given $k \ge0$ and $s =X_k$, there exists $n \in \mathbb{N}$ such that $nB \le k < (n+1)B$, and we have $\big [\!W_{(n+1)B} \cdots W_{k}\!\big] \!(s,s) \!\!\ge\!\! \left[\!\frac{1}{\vert \mathcal{S} \vert} \!\right]^{(n\!+\!1)\!B\!-\!k}\! \!\ge\!\! \left[\!\frac{1}{\vert \mathcal{S} \vert} \!\right]^B$.

% \vspace{-0.5em}
Let $\mathcal{N}^s(l)$ denote the supervisors that can be reached with positive probability at time step $(\!n\!+\!1\!)\!B\!+\!l$.
It follows that $\mathcal{N}^s\!(l) \!\!\subseteq \!\!\mathcal{N}^s(l\!+\!1)$ since $W_l(\!s'\!,s')\!\!>\!\!0$ for all $s' \!\!\in\!\! \mathcal{N}^s(l)$.
If $\mathcal{N}^s\!(n'\!B) \!\!\neq\!\! \mathcal{S}$, there exists $n'\!B\! \le\!\! l'\!\! <\!\! (n'\!+\!1)B$, $s \!\!\in\!\! \mathcal{N}^s\!(n'\!B)$ and $s' \!\!\notin\!\! \mathcal{N}^s(n'\!B)$ such that $\left\{ \!s,s' \!\right\} \!\!\in\! \!E_l$ due to Assumption \ref{dynamic_assumption}.
Hence, we obtain $\mathcal{N}^s\!(n'\!B)\!\! \subsetneq \!\!\mathcal{N}^s(l') \!\!\subseteq\!\! \mathcal{N}^s\big(\!(\!n'\!+\!1\!)\!B \!\big)$.
Then we can conclude that $\mathcal{N}^s\!(\vert \mathcal{S} \vert B)\!\!=\!\!\mathcal{S}$.
Let $L\!\!=\!\!(\vert \mathcal{S} \vert\!+\!1)\! B$, then we have
$
\mathbb{P}\!\big( \!X_{k\!+\!L}\!\!=\!\!s' \big\vert \!X_k\!\!=\!\!s \!\big)\!\! > \!\!0
$.
By the definition of $W_{l'}$, $W_l(s,s')\!\!>\!\!0$ implies $W_{l'}(s,s')\!\!\ge\!\! \frac{1}{\vert \mathcal{S} \vert}$.
It follows that there exists at least one path satisfying $\mathbb{P}\big( \!X_{k+L}\!\!=\!\!s' \big\vert\! X_k\!\!=\!\!s \!\big)\!\! > \!\!0$, 
namely, 
$\mathcal{P}_{\!s,s'\!,L}\!\!=\!\!\big(\!X_k\!\!=\!\!s,\!\dots\!,X_{l'}\!\!=\!\!s,X_{l'\!+\!1}\!\!=\!\!s'\!,\!\dots\!,X_{k\!+\!L}\!\!=\!\!s'\big)$.
Thus, we have $\mathbb{P}\big(\! X_{k\!+\!L}\!\!=\!\!s' \!\big\vert\! X_k\!\!=\!\!s \!\big) \!\!\ge\!\! \mathbb{P} \!\left(\mathcal{P}_{\!s,s'\!,L}\right) \!\!\ge\!\! \left[\!\frac{1}{\vert \mathcal{S} \vert} \!\right]^{\!L}$.
This completes the proof.
\end{proof}
Now we can prove Theorem \ref{belief_convergence_dynamic}.
We first define the error of $\pi_k^{m,s}$ as the difference between $\pi_k^{m,s}$ and $\pi_k^{m}$:
% \vspace{-3em}
\begin{align}
\label{supervisor_belief_error}
e_k^{m,s}=\pi_k^{m,s}-\pi_k^m,\quad \forall m \in \mathcal{T},\ \forall s \in \mathcal{S},
\end{align}
% \vspace{-2em}
and then we analyze the decay rate of $e_k^{m,s}$.
Given a team $m$, we introduce the following notation: $D^m=\mathrm{diag}\big(ST(m,1),\dots,ST(m,\vert \mathcal{S} \vert)\big)$.
Then we obtain
% \vspace{-3em}
\begin{align}
\label{supervisor_belief_update}
\left (\!\!\begin{array}{c}
\pi_{k+1}^{m,1}\\[-3pt]
\vdots\\[-3pt]
\pi_{k+1}^{m,\vert \mathcal{S} \vert}\\
\end{array}\!\!\right)\!\!=\!D^m\!\!
\left (\!\!\begin{array}{c}
\pi_{k+1}^{m}\\[-3pt]
\vdots\\[-3pt]
\pi_{k+1}^{m}\\
\end{array}\!\!\right)\!\!+\!(I\!-\!D^m) W_k \!
\left (\!\!\begin{array}{c}
\pi_{k}^{m,1}\\[-3pt]
\vdots\\[-3pt]
\pi_{k}^{m,\vert \mathcal{S} \vert}\\
\end{array}\!\!\right).
\end{align}
% \vspace{-2em}
Since $\Vert W_k(s,\cdot) \Vert_1>0$, combining (\ref{supervisor_belief_error}) and (\ref{supervisor_belief_update}), we obtain
% \vspace{-3em}
\begin{align}
&
\left (\begin{array}{c}
e_{k+1}^{m,1}\\[-3pt]
\vdots\\[-3pt]
e_{k+1}^{m,\vert \mathcal{S} \vert}\\
\end{array}\right)
=
\left (\begin{array}{c}
\pi_{k+1}^{m,1}\\[-3pt]
\vdots\\[-3pt]
\pi_{k+1}^{m,\vert \mathcal{S} \vert}\\
\end{array}\right)
-
\left (\begin{array}{c}
\pi_{k+1}^{m}\\[-3pt]
\vdots\\[-3pt]
\pi_{k+1}^{m}\\
\end{array}\right)\nonumber\\ 
% =&(I\!-\!D^m) W_k \!\!
% \left[\!\left (\!\begin{array}{c}
% e_{k}^{m,1}\\
% \vdots\\
% e_{k}^{m,\vert \mathcal{S} \vert}\\
% \end{array}\!\right)\!\!+\!\!\left (\!\begin{array}{c}
% \pi_k^m\\
% \vdots\\
% \pi_k^m\\
% \end{array}\!\right)\!\right]\!+\!(D^m\!-\!I)\!\!
% \left (\!\begin{array}{c}
% \pi_{k+1}^{m}\\
% \vdots\\
% \pi_{k+1}^{m}\\
% \end{array}\!\right)\nonumber\\
=&(I\!-\!D^m) W_k \!
\left (\!\!\begin{array}{c}
e_{k}^{m,1}\\[-3pt]
\vdots\\[-3pt]
e_{k}^{m,\vert \mathcal{S} \vert}\\
\end{array}\!\!\right) \!+\!(I\!-\!D^m)\!\!
\left (\!\!\begin{array}{c}
\pi_k^m-\pi_{k+1}^{m}\\[-3pt]
\vdots\\[-3pt]
\pi^m_k-\pi_{k+1}^{m}\\
\end{array}\!\!\right)\!.
\label{matrix_update_of_belief_error}
\end{align}
% \vspace{-2em}

Let $e_k^m \!\!\triangleq\!\!\! \left (\!\!\!\begin{array}{c}
e_{k}^{m,1}\\[-5pt]
\vdots\\[-5pt]
e_{k}^{m,\vert \mathcal{S} \vert}\\
\end{array}\!\!\!\right)$, $b_{k+1}^m \!\!\triangleq\!\! (I-D^m)\!\!
\left (\!\!\!\begin{array}{c}
\pi_{k}^m\!-\!\pi_{k+1}^{m}\\[-5pt]
\vdots\\[-5pt]
\pi^m_{k}\!-\!\pi_{k+1}^{m}\\
\end{array}\!\!\!\right)$ and $B^m_k \!\!\triangleq\!\! (I\!-\!D^m)W_k$.
Then we rewrite (\ref{matrix_update_of_belief_error}) as $e_{k+1}^m\!=\!B^m_k e_k^m\!+\!b_{k+1}^m$.
% \begin{align*}
% % \label{supervisor_belief_error_update}
%     e_{k+1}^m=B^m_k e_k^m+b_{k+1}^m.
% \end{align*}
Using the above recursion, we obtain $e_k^m\!=\!\mathcal{B}^m(k,0)e_0^m\!+\!\sum_{l=0}^{k-1}\mathcal{B}^m(k,l\!+\!1)b_l^m$, where $\mathcal{B}^m(k,l)\!=\!B^m_{k-1} \cdots B_{l}^m$.
By reordering the indices, $W_k^m$ and $D^m$ can be written as 
% \vspace{-2.8em}
% \begin{align*}
%     e_k^m=\mathcal{B}^m(k,0)e_0^m+\sum_{l=0}^{k-1}\mathcal{B}^m(k,l+1)b_l^m,
% \end{align*}
% \vspace{-1.8em}
% \vspace{-3em}
\begin{align*}
    W^m_k=
\left (\begin{array}{cc}
W_{k,U^mU^m} & W_{k,U^mO^m}\\
W_{k,O^mU^m} & W_{k,O^mO^m}\\
\end{array}\right),\ \tilde{D}^m=
\left (\begin{array}{cc}
0 & 0\\
0 & I\\
\end{array}\right).
\end{align*}
% \vspace{-2em}
Then we have
% \vspace{-3em}
\begin{align*}
    &B_k^m\!\!\! =\!\!\! \left (\!\!\!\begin{array}{cc}
    W_{k,U^m\!U^m} & \!\!W_{k,U^m\!O^m}\\
    0 & 0\\
    \end{array}\!\!\!\!\right)\!\!,
    \mathcal{B}^m\!(\!k,\!l)\!\!=\!\!\! \left (\!\!\!\begin{array}{cc}
    B^m_{U\!U}\!(\!k,\!l) & \!\!B^m_{U\!O}\!(\!k,\!l)\\
    0 & 0\\
    \end{array}\!\!\!\right)\!\!,
\end{align*}
% \vspace{-2em}
where $B^m_{U\!U}\!(k,l)\!\!=\!\!W_{k\!-\!1,U^m\!U^m} \!\cdots\! W_{l,U^m\!U^m}$ and $B^m_{UO}(k,l)\!\!=\!\! W_{k-1,U^mO^m} \cdots W_{l,U^mO^m}$.
We denote $\Vert\! W^m\!(k,\!l) \!\Vert_{\infty}\!\! =\!\!R^m\!(\!k,\!l)\!\!=\!\!\mathop{\max}_{s \in U^m} \!\!\sum_{s' \!\in\! U^m}\!\!\left[W^m\!(\!k,\!l)\right]\!(\!s,\!s')$ and $W^m(k,l)\!\!=\!\!W_{\!k\!-\!1,U^{\!m}\!U^{\!m}} \!\cdots\! W_{\!l,U^{\!m}\!U^{\!m}}$.
By Lemma \ref{dynamic_path_lemma}, we have $R^m\!(l\!+\!L,l) \!\!\le\!\! 1\!\!-\!\!\alpha$ with $\alpha\!\!=\!\!\!\left[\!\frac{1}{\vert \mathcal{S} \vert} \!\right]^{\!\!L}$.
Thus, we get $ R^m\!(k,l)\!\! \le\!\! \alpha^{[\frac{k\!-\!l}{L}]} \!\!\le\!\! C'\! \rho^{k\!-\!l}$, where $C'\!\!=\!\!(1\!\!-\!\!\alpha)^{\!-\!1}\!\!>\!\!1$ and $\rho\!\!=\!\!(1\!\!-\!\!\alpha)^{\!\frac{1}{L}}\!\!<\!\!1$.
For any $x\!\!=\!\!(x_{U^m},x_{O^m})^T \!\!\in\!\! \mathbb{R}^m$, we get $B^m_l x\!\!=\!\! (y_{U^m},0)^T$.
Then we obtain $\mathcal{B}^m(k,l)x\!\!=\!\!(W^m(k,l\!+\!1)y_{U^m},0)^T$ and have
% \vspace{-3em}
\begin{align}
\label{1}
\Vert \mathcal{B}^m(k,l)x \Vert_{\infty} &\le \Vert W^m(k,l+1)\Vert_{\infty} \Vert y_{U^m} \Vert_{\infty} \nonumber\\[-3pt]
&\le C'\!\rho^{k-l-1}\Vert y_{U^m} \Vert_{\infty}.
\end{align}

% \vspace{-2em}
From $y_{U^m}\!\!=\!\!W_{l,U^mU^m}x_{U^m}\!\!+\!\!W_{l,U^mO^m}x_{O^m} \!\!\triangleq\!\! M(l)x$ and $\Vert\! M(l) \!\Vert_{\infty} \!\!\le\!\! 1$, we get that $\Vert y_{U^m} \Vert_{\infty} \!\!\le\!\! \Vert x \Vert_{\infty}$, combing which and (\ref{1}), we have $\Vert \mathcal{B}^m(k,l) \!\Vert_{\infty} \!\!\le \!\!C''\! \rho^{k-l}$.
Hence, we obtain
% \vspace{-3em}
\begin{align*}
    &\Vert e_k^m \Vert_{\infty}
    =\Vert \mathcal{B}^m(k,0)e_0^m+\sum_{l=0}^{k-1}\mathcal{B}^m(k,l+1)b_l^m\Vert_{\infty} \\[-5pt]
    \le& C'' \rho^k \Vert e_0^m \Vert_{\infty}+\sum_{l=1}^{k-1}C''\rho^{k-l-1} \Vert b_l^m\Vert_{\infty}\\[-3pt]
    =& C'' \!\!\rho^k \Vert e_0^m\! \Vert_{\infty}\!\!+\!\!\!\sum_{l=1}^{[\frac{k}{2}]-1}\!\!C''\!\rho^{k-l-1} \Vert b_l^m\!\Vert_{\infty}\!\! +\!\!\!\sum_{l=[\frac{k}{2}]}^{k-1}\!\!C''\!\rho^{k-l-1} \!\Vert b_l^m\!\Vert_{\infty}\\[-3pt]
    \le& C''\!\! \rho^k\! \Vert e_0^m \!\Vert_{\infty}\!\!+\!C''\!\!\!\!\sum_{l=k-[\frac{k}{2}]}^{\infty}\!\!\!\rho^{l} \!\mathop{sup}_{l \ge0}\!\Vert b_l^m\!\Vert_{\infty}\!\!+\!C''\!\sum_{l=0}^{\infty}\!\rho^{l} \!\!\mathop{sup}_{l \ge [\frac{k}{2}]}\!\Vert b_l^m\!\Vert_{\infty}\\[-5pt]
    \le&  C''\!\!\rho^k \!\Vert e_0^m \!\Vert_{\infty}\!\!+\!C''\!\mathop{sup}_{l \ge0}\!\Vert b_l^m\!\Vert_{\infty}  \frac{\rho^{k-[\frac{k}{2}]}}{1-\rho}\!\!+\!C''\!\!\mathop{sup}_{l \ge [\frac{k}{2}]}\!\!\Vert b_l^m\!\Vert_{\infty}  \frac{1}{1-\rho}\\[-4pt]
    \le&  C'' \!\rho^k\! \Vert e_0^m \!\Vert_{\infty}\!\!+\!C''\!\mathop{sup}_{l \ge0}\!\Vert b_l^m\!\Vert_{\infty} \frac{\rho^{\frac{k}{2}}}{1-\rho}\!\!+\!C''\!\!\mathop{sup}_{l \ge [\frac{k}{2}]}\!\Vert b_l^m\!\Vert_{\infty} \frac{1}{1-\rho}\\[-4pt]
    =&  C''\! \rho^k\! \Vert e_0^m \!\Vert_{\infty}\!\!+\!C''\!\!\mathop{sup}_{l \ge0}\alpha_{l-1}\!\Vert (\!I\!\!-\!\!D^m\!)(\underline{a}_{l-1}^m\!\!-\!\!\pi_{l-1}^m\!) \Vert_{\infty} \frac{\rho^{[\frac{k}{2}]}}{1-\rho}\\[-5pt]
    &+\frac{C''}{1-\rho}\mathop{sup}_{l \ge [\frac{k}{2}]}\alpha_{l-1}\Vert (I-D^m)(\underline{a}_{l-1}^m-\pi_{l-1}^m) \Vert_{\infty}\\[-5pt]
    \le& C'' \!\rho^k \!\Vert e_0^m \!\Vert_{\infty}\!\!+\!\!2C''\alpha_{0}\cdot \frac{\rho^{[\frac{k}{2}]}}{1-\rho}+\frac{2C''}{1-\rho}\cdot \alpha_{[\frac{k}{2}]-1}\\[-5pt]
    =&O(\rho^{[\frac{k}{2}]})+O(\alpha_{[\frac{k}{2}]}).
\end{align*}
% Thus. we have
% \vspace{-2.8em}
% \begin{align*}
%     &\Vert e_k^m \Vert_{\infty}\\
%     &+\frac{C''}{1-\rho}\mathop{sup}_{l \ge [\frac{k}{2}]}\alpha_{l-1}\Vert (I-D^m)(\underline{a}_{l-1}^m-\pi_{l-1}^m) \Vert_{\infty}\\
%     \le& C'' \!\rho^k \!\Vert e_0^m \!\Vert_{\infty}\!\!+\!\!2C''\alpha_{0}\cdot \frac{\rho^{\frac{k}{2}}}{1-\rho}+\frac{2C''}{1-\rho}\cdot \alpha_{[\frac{k}{2}]-1}
% \end{align*}
% \vspace{-2em}
Finally, we can conclude $\Vert \pi_k^{m,s}-\pi_k^m \Vert_{\infty}\le O(\rho^{[\frac{k}{2}]})+O(\alpha_{[\frac{k}{2}]})$ and complete the proof of Theorem \ref{belief_convergence_dynamic}.

%% file: appendices_IEEE/Proof_of_gap.tex
\subsection{Proof of Theorem \ref{convergence_of_gap}}
\label{proof_of_gap}
The TNG of Algorithm \ref{alg:team_fp_with_decentralized_belief_learning} is closely related to the true belief $\pi_k^m$.
The proof of Theorem \ref{convergence_of_gap} proceeds in two steps: 1) constructing a reference scenario where true beliefs are frozen and showing that members within a team can learn to team up in spite of the errors caused by the supervisor networks; 2) bounding the TNG by leveraging stochastic differential inclusion approximations.

% \textbf{Reference scenario and error analysis.}
% We first introduce the reference scenario. 
% \vspace{-0.5em}
Given a team strategy profile $\pi\!\!= \!\!\big\{\! \pi^m\! \!\in\! \!\Delta\!\left(\underline{\mathcal{A}}^m\right) \!\!\big\}_{\!m \in \!\mathcal{T}}$, for any agent $i\!\! \in\!\! \mathcal{I}^m$ and any $a^{\!-\!i} \!\!\in\!\! \prod_{j \in \mathcal{I}^m\!, j\neq i} \!\mathcal{A}^j$, (\ref{potential_function_difference}) implies $br_{\!\tau} \!\big(\! u^i\!( \!\cdot,a^{\!-\!i}\!,\!\pi^{\!-\!m} \!)\! \big)\!\!=\!\!br_{\!\tau} \big(\! \phi^m\!(\! \cdot,a^{\!-\!i}\!,\!\pi^{\!-\!m} \!)\! \big)$, where $\pi^{\!-\!m} \!\!\triangleq \!\!\{ \pi^l \}_{\!l \neq m}$.
We divide the learning process into epochs of length $T$.
By accumulating the true belief update (\ref{belief_update}) from $nT$ to $(\!n\!+\!1\!)T\!-\!1$, we obtain
% \vspace{-3em}
\begin{align}
\label{belief_accumulation_1}
    \pi_{(\!n\!+\!1\!)T}^m 
    \!\!=\!\!\!\left[\! \prod_{k=nT}^{(\!n\!+\!1\!)T\!-\!1}\!\!\!\! (\!1\!\!-\!\!\alpha_k\!)\! \right] \!\!\pi_{nT}^m\!+\!\!\!\sum_{k=nT}^{(\!n\!+\!1\!)T\!-\!1} \!\!\!\!\!\alpha_k \!\!\left[ \!\prod_{l=k\!+\!1}^{(\!n\!+\!1\!)T\!-\!1} \!\!\!\!(\!1\!\!-\!\!\alpha_l\!) \!\right]\!\!\underline{a}_k^m\!,
\end{align}
% \vspace{-2em}
where $n\!\!=\!\!0,\!1,\!\dots$ is the epoch index.
Let $\pi^m_{(n)} \!\!\triangleq\!\! \pi^m_{nT}$ and $\pi^{m,s}_{(n)} \!\!\triangleq\!\! \pi^{m,s}_{nT}$ denote the true belief and the actual belief about team $m$ after learning $n$ epochs.
We further define
% \vspace{-3em}
\begin{align}
\label{definitions_of_beta}
    \beta_k \triangleq \alpha_k \prod_{l=k+1}^{(n+1)T-1} \!\!(1-\alpha_l)\ \text{ and }\ \beta_{(n)} \triangleq \sum_{k=nT}^{(n+1)T-1} \!\!\beta_k.
\end{align}
% \vspace{-2em}
Then, we rewrite (\ref{belief_accumulation_1}) as
% \vspace{-3em}
\begin{align}
\label{belief_accumulation_2}
    \pi_{(n\!+\!1)T}^m\!=\!\left(\!1\!-\!\beta_{(n)} \!\right)\! \cdot\! \pi_{(n)}^m\!+\!\beta_{(n)} \!\!\left[ \sum_{k=nT}^{(n\!+\!1)T\!-\!1} \!\!\frac{\beta_k}{\beta_{(n)}} \cdot \underline{a}_k^m\right].
\end{align}
% \vspace{-2em}
By Assumption 3 and Lemma 5.4 of \cite{donmez2024team}, we obtain
% \vspace{-3em}
\begin{align}
\label{beta_properties}
\beta_{(n)}\! \!\to\! 0 \text{ as } n\! \!\to\! \infty,\ 
\sum_{n=0}^{\infty}\! \beta_{(n)}\!=\!\infty \ \text{and}\ 
\sum_{n=0}^{\infty} \!\beta_{(n)}^2\! \!< \!\infty.
\end{align}

% \vspace{-2em}
Let $\mathcal{F}_{(n)}$ denote the filtration generated by the $\sigma$-algebra $\sigma(\mathcal{S},ST,\{SN_k\}_{k \ge 0},A_0,\dots,A_{nT-1})$, where $A_t$ denotes the joint action profile of all teams at round $t$, i.e., $A_t=(\underline{a}_t^1,\dots,\underline{a}_t^{|\mathcal{T}|})$.
Note that $\pi_{(n)}^m$ and $\pi_{(n)}^{m,s}$ are $\mathcal{F}_{(n)}$-measurable.
The joint action distributions of team $m$ based on the true beliefs, i.e., in the ideal scenario, are defined for DTOA and iDTOA at time $k$ in epoch $n$ as
% \vspace{-3em}
\begin{align*}
    \text{team-}\nu_{(n),k}^m \triangleq \mathbb{E} \left[ \underline{a}_k^m \big| \mathcal{F}_{(n)} \right],\ \text{indp-}\nu_{(n),k}^m \triangleq \mathbb{E} \left[ \underline{a}_k^m \big| \mathcal{F}_{(n)} \right]
\end{align*}
% \vspace{-2em}
for $k=nT,\dots,(n\!+\!1)T\!-\!1$.
We also define the distributions of the joint action of team $m$ based on the actual beliefs at round $k$ in epoch $n$ for DTOA and iDTOA as
% \vspace{-3em}
\begin{align*}
    \text{a-}\text{team-}\nu_{(n),k}^m \!\triangleq\! \mathbb{E} \left[ \underline{a}_k^m \big| \mathcal{F}_{(n)} \right]\!\!,\text{a-}\text{indp-}\nu_{(n),k}^m \!\triangleq\! \mathbb{E} \left[ \underline{a}_k^m \big| \mathcal{F}_{(n)} \right]
\end{align*}
% \vspace{-2em}
for $k\!\!=\!\!nT\!,\dots,\!(\!n\!+\!1\!)T\!\!-\!\!1$.
For notational simplicity, we use $\nu_{(n),k}^m$ to represent $\text{team-}\nu_{(n),k}^m$ for DTOA and $\text{indp-}\nu_{(n),k}^m$ for iDTOA, and use $\text{a-}\nu_{(n),k}^m$ to represent $\text{a-}\text{team-}\nu_{(n),k}^m$ for DTOA and $\text{a-}\text{indp-}\nu_{(n),k}^m$ for iDTOA.

Then, we rewrite (\ref{belief_accumulation_2}) in form of stochastic approximation:
% \vspace{-3em}
\begin{align}
\label{belief_accumulation_3}
    \pi_{(\!n\!+\!1\!)}^m
    \!\!=\!\!\left(\! 1\!\!-\!\!\beta_{(\!n\!)} \!\right)\!\pi_{(\!n\!)}^m \!\!+\!\!\beta_{(\!n\!)}\!\!\! \left[\! \sum_{k=nT}^{(\!n\!+\!1\!)T\!-\!1} \!\!\!\!\frac{\beta_k}{\beta_{(\!n\!)}}\! \nu_{(\!n\!),k}^m\!\!+\!\text{a-}\omega_{(\!n\!+\!1\!)}^m \!\right]\!\!,
\end{align}
% \vspace{-2em}
where $\text{a-}\omega_{(n+1)}^m$ is defined as
% \vspace{-3em}
\begin{align}
\label{decomposition_of_omega}
    \text{a-}\omega_{(n+1)}^m \triangleq &\sum_{k=nT}^{(n+1)T-1} \frac{\beta_k}{\beta_{(n)}} \left[ \underline{a}_k^m-\nu_{(n),k}^m \right] \nonumber\\[-1pt]
    =&\!\!\!\sum_{k\!=\!nT}^{(\!n\!+\!1\!)\!T\!-\!1} \!\!\!\!\frac{\beta_k}{\beta_{(\!n\!)}}\!\! \left[ \underline{a}_k^m\!\!-\!\!\text{a-}\nu_{(\!n\!),k}^m \!\right] \!\!+\!\!\!\!\!\!\sum_{k\!=\!nT}^{(\!n\!+\!1\!)\!T\!-\!1} \!\!\!\!\frac{\beta_k}{\beta_{(\!n\!)}} \!\!\left[ \text{a-}\nu_{(n),k}^m\!\!-\!\!\nu_{(\!n\!),k}^m\! \right] \nonumber\\[-1pt]
    =&\omega_{(n+1)}^m+\text{a-}e_{(n+1)}^m,
\end{align}
% \vspace{-2em}
where $\omega_{(\!n\!+\!1\!)}^m\!\!\triangleq\!\! \sum_{\!k\!=\!nT}^{\!(\!n\!+\!1\!)T\!-\!1} \!\!\frac{\beta_k}{\beta_{(\!n\!)}} \!\!\left[ \!\underline{a}_k^m\!\!-\!\!\text{a-}\nu_{(\!n\!)\!,k}^m \!\right]$ is a martingale difference sequence and we define the actual error as $\text{a-}e_{(n\!+\!1)}^m \!\!\triangleq\!\! \sum_{k=nT}^{(n\!+\!1)T\!-\!1} \!\!\frac{\beta_k}{\beta_{(n)}} \!\left[ \!\text{a-}\nu_{(n),k}^m\!\!-\!\!\nu_{(n),k}^m \!\right]$.
It follows that $\mathbb{E} \!\left[ \omega_{(n+1)}^m \big \vert \mathcal{F}_{(n)} \right ]\!\!\!\!=0$.
Then we are ready to prove Lemma \ref{distribution_error}.
Recall the result in Lemma \ref{distribution_error}: the difference between $\text{a-}\nu_{(n),k}^m$ and $\nu_{(n),k}^m$ can be bounded as
%\begin{align*}
    $\left\Vert \text{a-}\nu_{(n),k}^m\!\!-\!\!\nu_{(n),k}^m \right\Vert_1\!\! \le\!\! C \delta_{(n)}$,
%\end{align*}
where $\delta_{(n)}\!\! \rightarrow \!0$ as $n \!\!\rightarrow\! \infty$.

% \vspace{-1em}
\begin{proof}
    We first define the joint-action process from the start of epoch $n$.
    We denote the joint-action profiles of all teams by $\{ {\omega}_k \}_{k \ge nT} \triangleq \left\{ \left( \underline{a}_k^m \right)_{m \in \mathcal{T}} \big|\mathcal{F}_{(n)} \right\}$, and view each joint-action profile ${\omega}_k$ as a state.
    The transitions between states depend only on the beliefs about all teams, $ST$, and $\{ SN_k \}_{k \ge 0}$.
    Let $P_{(n),k}$ denote the transition probabilities between states and let $a\text{-}\nu_{(n),k}$ denote the state distribution at round $k$ in the actual scenario.
    Let ${P}^*_{(n)}$ denote the transition probabilities between states and let $\nu_{(n),k}$ denote the state distribution at round $k$ in the ideal scenario.
    Then, the joint-action distributions of team $m$ at round $k$ within epoch $n$ in the actual and ideal scenarios are defined as
    % \vspace{-3em}
    \begin{align}
    \label{wo_yao_yong}
    a\text{-}\nu_{(\!n\!),k}^m(a)\!\!=\!\!\!\!\!\!\sum_{\omega:a^m\!=\!a}\!\!\!a\text{-}\nu_{(\!n\!),k}(\omega),
    \nu_{(\!n\!),k}^m(a)\!\!=\!\!\!\!\!\!\sum_{\omega:a^m\!=\!a}\!\!\!\nu_{(\!n\!),k}(\omega).
    \end{align} 

    % \vspace{-2em}
    Since $\phi^m$ is linear in $\pi$, the smoothed best response is Lipschitz in $\pi$.
    Thus, there exists $L_{\pi}>0$ such that
    % \vspace{-3em}
    \begin{align*}
        &\left\Vert \!P_{(n),k}\! \left( \cdot \big| \omega_k, \!\dots\!, \omega_{nT};\!\pi_k^{ \mathcal{S}} \right)\!\!-\!\!P^*_{(n),k} \!\left( \cdot \big| \omega_k, \!\dots\!, \omega_{nT};\!\pi_k \right)\! \right\Vert_{TV} \nonumber\\[-2pt]
        \le& L_{\pi} \sum_{m \in \mathcal{T}} \mathop{\max}_{s \in \mathcal{S}} \Vert \pi_{k}^{m,s}-\pi_{k}^m \Vert_1,
    \end{align*}
    % \vspace{-2em}
    where $\pi_k^{\mathcal{S}}$ denotes the belief estimates provided by the supervisor network, based on which agents choose their actions.
    Consider the error $ D_k \!\!\triangleq\!\! \left\Vert \text{a-}\nu_{(n),k}\!\!-\!\!{\nu}_{(n),k} \right\Vert_1$.
    Since $\text{a-}\nu_{(\!n\!),k\!+\!1}\!\!=\!\!\text{a-}\nu_{(\!n\!),k}P_{(\!n\!),k}$ and ${\nu}_{(\!n\!),k\!+\!1}\!\!=\!\!{\nu}_{(\!n\!),k}P_{(\!n\!),k}^*$, we have 
    % \vspace{-3em}
    \begin{align*}
    % \label{fen_bu_cha_ju_1}
        D_{k+1}=&\Vert \text{a-}\nu_{(n),k}P_{(n),k}- {\nu}_{(n),k}P_{(n),k}^* \Vert_1\nonumber\\[-2pt]
        =&\bigg\Vert \!\!\left( \text{a-}\nu_{(n),k}\!\!-\!\!{\nu}_{(n),k} \right)\!\!P_{(n),k}^*\!\!+\!\text{a-}\nu_{(n),k} \!\!\left( \!\!P_{(n),k}\!\!-\!\!P_{(n),k}^* \! \right) \!\!\bigg\Vert_1\!.
    \end{align*}
    % \vspace{-2em}
    Using the fact that $\Vert xP \Vert_1 \le \Vert x \Vert_1$ and the above decomposition, we obtain 
    % \vspace{-3em}
    \begin{align*}
        D_{k+1} \le& \left\Vert \!\left( \text{a-}\nu_{(\!n\!),k}\!-\!{\nu}_{(\!n\!),k} \right)\!\!P_{(\!n\!),k}^*\!\right\Vert_1\!\!+\!\!\left\Vert \text{a-}\nu_{(\!n\!),k} \!\!\left(\!\!P_{\!(\!n\!),k}\!\!-\!\!P_{\!(\!n\!),k}^* \! \right)\! \right\Vert_1\\[-2pt]
        \le &D_k+\left\Vert \text{a-}\nu_{(n),k} \left(P_{(n),k}-P_{(n),k}^*  \right) \right\Vert_1\\[-2pt]
        \le &D_k+2 \mathop{\sup}_{\omega} \Vert P_{(n),k}(\omega,\cdot)-P_{(n),k}^*(\omega,\cdot) \Vert_{TV}\\[-3pt]
        \le &D_k+2L_{\pi} \sum_{m \in \mathcal{T}} \mathop{\max}_{s \in \mathcal{S}} \Vert \pi_{k}^{m,s}-\pi_{k}^m \Vert_1.
    \end{align*}

    % \vspace{-2em}
    By the definition of the maximum belief-estimation error in epoch $n$, $\delta_{(\!n\!)}\!\!=\!\!\mathop{\max}_{k \in\! [nT\!\!-\!\!1,(\!n\!+\!1\!)T\!-\!1]}\!\!\sum_{m \in \mathcal{T}} \!\mathop{\max}_{s \in \mathcal{S}} \!\!\Vert \pi_{k}^{m,s}\!\!-\!\!\pi_{k}^m \Vert_1$, we have $D_{k} \le D_{nT}+2TL_{\pi} \delta_{(n)}$.
   Moreover, $ D_{nT} \!\le \!\sum_{m \in \mathcal{T}} C'\Vert \pi_{nT-1}^{m,s}\!-\!\pi_{nT-1}^{m} \Vert_1$.
    Thus we have $D_k \le C \delta_{(n)}$, where $C\!=\!C'\!+\!2TL_{\pi}$.
    From (\ref{wo_yao_yong}), we have
    % \vspace{-3em}
    \begin{align}
        \Vert \text{a-}\nu_{(n),k}^m\!\!-\!\!\nu_{(n),k}^m \Vert_1 \!\!\le \!\!\Vert \text{a-}\nu_{(n),k}\!\!-\!\!\nu_{(n),k} \Vert_1 \!\!=\!\!\! D_k \!\!\le\!\! C \delta_{(n)}.
    \end{align}
    % \vspace{-2em}
    This completes the proof of Lemma \ref{distribution_error}.
\end{proof}

% \vspace{-1em}
Finally, we can get $\Vert \text{a-}e_{(n\!+\!1)}^m\Vert_1 \!\le\! \sum_{k=nT}^{(n+1)T-1} \!\frac{\beta_k}{\beta_{(n)}} C \delta_{(n)}\!\!=C \delta_{(n)}$.
Furthermore, $\delta_{(n)}\to0$ as $n\to\infty$ by Theorem \ref{belief_convergence_dynamic}.
Now we introduce a reference scenario to facilitate the analysis.
Let $\hat{\pi}_t^m$ denote the belief for team $m$ at round $t$ in the reference scenario.
In the reference scenario, $\hat{\pi}_t^m$ is only updated at the end of each epoch.
In other words, for all $nT \!\le\! t \!\le\! (n\!+\!1)T\!-\!1$ and $m \!\in\! \mathcal{T}$, we have $\hat{\pi}_t^m\!=\!\pi_{(n)}^m$.
Since the beliefs are fixed, DTOA dynamics reduce to log-linear learning in the reference scenario.
Let $\hat{a}^m_{(n),k}$ denote the joint action of team $m$ at round $k$ under the fixed beliefs in the reference scenario.

% \vspace{-0.5em}
Due to the nature of log-linear learning, $\{ \hat{a}_{(n),k}^m \}_{k=nT}^{\infty}$ forms a homogeneous Markov chain (MC).
In contrast, the actual action profiles $\{ \underline{a}_k^m \}_{k=nT}^{\infty}$ do not.
We define the stationary distributions of the MC in the reference scenario as $\check{\nu}_{(n),*}^m$ and $\hat{\nu}_{(n),*}^m$ for DTOA and iDTOA, respectively.
By \cite{blume1993statistical,marden2012revisiting}, it follows that $\check{\nu}_{(n),*}^m\!\!=\!br_{\!\tau} \!\!\left(\! \phi^m(\cdot,\pi_{(n)}^{-\!m})\!\right)$.
Therefore, we write the true belief update (\ref{belief_accumulation_3}) as
% \vspace{-3em}
\begin{align}
\label{belief_accumulation}
    \pi_{\!(\!n\!+\!1\!)}^{\!m}\!\!=\!\!\left( \!1\!\!-\!\!\beta_{(\!n\!)} \!\right)\! \pi_{\!(\!n\!)}^{\!m} 
   \!\! +\!\! \beta_{(\!n\!)}\!\!  \left[\! br_{\!\tau}\!\!\left( \!\!\phi^{m}\!(\!\cdot,\!\pi_{\!(\!n\!)}^{\!\!-\!m}\!)\!\! \right)\!\!+\!\!\omega_{\!(\!n\!+\!1\!)}^{\!m}\!\!\!+\!\!e_{\!(\!n\!)}^{\!m}\!\!\!+\!\!a\text{-}e_{\!(\!n\!+\!1\!)}^{\!m} \!\right]\!\!.
\end{align}
% \vspace{-2em}
The error for DTOA is
% \vspace{-3em}
\begin{align}
\label{error_for_team_FP}
    \text{team-}e_{(n)}^m
    % &= \sum_{k=nT}^{(n+1)T-1}\!\! \frac{\beta_k}{\beta_{(n)}} \cdot \text{team-}\nu_{(n),k}^m\!-\!\check{\nu}_{(n),*}^m \nonumber\\[-3pt]
    &=\sum_{k=nT}^{(n+1)T-1} \!\!\frac{\beta_k}{\beta_{(n)}}\left( \text{team-}\nu_{(n),k}^m\!-\!\check{\nu}_{(n),*}^m \right).
\end{align}
% \vspace{-2em}
The error for iDTOA is $\text{indp}\text{-}e_{(n)}^m \!\!\triangleq\!\! \hat{e}^m_{(n)}\!\!+\!\!\check{e}^m_{(n)}$, where $\check{e}^m_{(n)} \!\!\triangleq\!\! \hat{\nu}_{(n),*}^m\!\!-\!\!\check{\nu}_{(n),*}^m$, and $\hat{e}^m_{(n)}\!\!= \!\!\sum_{k=nT}^{(n\!+\!1)T\!-\!1} \!\frac{\beta_k}{\beta_{(n)}} \!\left(\! \text{indp-}\nu_{(n),k}^m\!\!-\!\!\hat{\nu}_{(n),*}^m \!\right)$.
% \begin{align*}
%     \hat{e}^m_{(n)}= \sum_{k=nT}^{(n+1)T-1} \frac{\beta_k}{\beta_{(n)}} \left( \text{indp-}\nu_{(n),k}^m-\hat{\nu}_{(n),*}^m \right).
% \end{align*}
% \begin{align*}
%     \check{e}^m_{(n)} \triangleq \hat{\nu}_{(n),*}^m-\check{\nu}_{(n),*}^m.
% \end{align*}
% Note that $\nu_{(n),k}^m$ in (\ref{error_for_team_FP}) is based on the case of Team-FP, while $\nu_{(n),k}^m$ in (\ref{error_for_independent_team_FP}) is based on the case of Independent Team-FP.

\begin{lemma}
\label{bound_on_error}
    There exist constants $c,d,\rho \ge 0$ such that
    % \vspace{-3em}
    \begin{align*}
        \Vert \text{team-}\nu_{(n),k}^m\!-\!\check{\nu}_{(n),*}^m \Vert_1 \!\le\! c\rho^{k-nT}\!+\!dT\alpha_{nT},\ \forall m \in \mathcal{T},\\[-2pt]
        \Vert \text{indp-}\nu_{(n),k}^m\!-\!\hat{\nu}_{(n),*}^m \Vert_1 \!\le\! c\rho^{k-nT}\!+\!dT\alpha_{nT},\ \forall m \in \mathcal{T}.
    \end{align*}
\end{lemma}

% \vspace{-1.8em}
The proof of Lemma \ref{bound_on_error} can be found in Appendix B.1 of \cite{donmez2024team}.
Let $\text{team}\text{-}e_{(n)}^m$ and $\text{indp}\text{-}e_{(n)}^m$ denote the errors for DTOA and iDTOA, respectively.
Lemma \ref{bound_on_error} yields that $\text{team-} e_{(n)}^m$ and $\hat{e}_{(n)}^m$ can be bounded by
% \vspace{-3em}
\begin{align}
    &\Vert \text{team-}e_{(n)}^m \Vert_1 \le c \sum_{k=nT}^{(n+1)T-1} \frac{\beta_k}{\beta_{(n)}} \rho^{k-nT}+dR\alpha_{nT},\label{bound_on_team_error}\\[-2pt]
    &\Vert \hat{e}_{(n)}^m \Vert_1 \le c \sum_{k=nT}^{(n+1)T-1} \frac{\beta_k}{\beta_{(n)}} \rho^{k-nT}+dR\alpha_{nT}.\label{bound_on_indp_error_1}
\end{align}
% \vspace{-2em}
With Assumption \ref{step_size_assumption} and (\ref{definitions_of_beta}), we have $\frac{\beta_{k\!+\!1}}{\beta_k} \!\!\le\!\! 1$.
Thus,
% \vspace{-3em}
\begin{align}
\label{bound_on_beta_over_beta}
    \frac{\beta_k}{\beta_{(n)}} \le \frac{\beta_{nT}}{T\beta_{(n+1)T-1}} \le \frac{\alpha_{nT}}{T\alpha_{(n+1)T}}.
\end{align}
% \vspace{-2em}
By Assumption \ref{step_size_assumption}, we obtain
% \vspace{-3em}
\begin{align}
\label{limitation_of_alpha_over_alpha}
    \lim_{n \to \infty} \frac{\alpha_{nT}}{T\alpha_{(n+1)T}}=\frac{1}{T} \lim_{n \to \infty} \prod_{k=nT}^{(n+1)T-1} \frac{\alpha_k}{\alpha_{k+1}}=\frac{1}{T}.
\end{align}
% \vspace{-2em}
Using (\ref{bound_on_team_error}), (\ref{bound_on_indp_error_1}), (\ref{bound_on_beta_over_beta}), (\ref{limitation_of_alpha_over_alpha}), and the decay property of $\alpha_k$, we obtain, for all $m \in \mathcal{T}$,
% \vspace{-3em}
\begin{align*}
    \limsup_{n \to \infty}\! \left\Vert \text{team-}e_{(\!n\!)}^m \!\right\Vert_1 \!\!\le\!\! \frac{c}{T} \frac{1}{1\!\!-\!\!\rho}\text{ and }\limsup_{n \to \infty}\! \left\Vert \hat{e}_{(\!n\!)}^m \!\right\Vert_1 \!\!\le\!\! \frac{c}{T} \frac{1}{1\!\!-\!\!\rho}.
\end{align*}
% \vspace{-2em}

Let $\hat{\nu}$ and $\check{\nu}$ denote the stationary distributions induced by the classical and independent log-linear learning, respectively.
A small $\delta>0$ implies close stationary distributions in the classical and independent settings:
% \vspace{-3em}
\begin{align}
\label{bound_on_difference_between_classical_and_independent}
   \Vert  \hat{\nu}-\check{\nu} \Vert_1 \le \Lambda(\delta,\varepsilon_{\phi}),
\end{align}
% \vspace{-2em}
for some function $\Lambda$, and the difference $\Vert  \hat{\nu}-\check{\nu} \Vert_1$ decays to zero as $\delta \rightarrow 0^+$ for any $\varepsilon_{\phi}>0$.
Based on (\ref{bound_on_difference_between_classical_and_independent}), we can bound $\check{e}_{(n)}^m$ as $\check{e}_{(n)}^m \le \Lambda(\delta,\varepsilon)$ for some $\varepsilon>0$.
Hence, given $\varepsilon>0$, there exists $N_{\varepsilon} \in \mathbb{N}_+$ such that for any $ n > N_{\varepsilon}$,
% \vspace{-3em}
\begin{equation}
\begin{aligned}
    &\Vert \text{team-}e_{(n)}^m \Vert_1 < C_{team}(\varepsilon,T),\\
    & \Vert \text{indp-}e_{(n)}^m \Vert_1 < C_{indp}(\varepsilon,T),\Vert \text{a-}e_{(n+1)}^m\Vert_1 \le \varepsilon,
\end{aligned}
\label{11111}
\end{equation}
% \begin{align}
%     &\Vert \text{team-}e_{(n)}^m \Vert_1 < C_{team}(\varepsilon,T),\label{conclusion_bound_on_team_error}\\
%     &\Vert \text{indp-}e_{(n)}^m \Vert_1 < C_{indp}(\varepsilon,T),\label{conclusion_bound_on_indp_error}\\
%     &\Vert \text{a-}e_{(n+1)}^m\Vert_1 \le \varepsilon,\label{11111}
% \end{align}
% \vspace{-2em}
where $C_{team}(\varepsilon,T)\!\!=\!\!\varepsilon\!+\!\frac{c}{T}\! \frac{1}{1-\rho}$ and $C_{indp}(\varepsilon,T)\!\!=\!\!\varepsilon\!+\!\frac{c}{T} \!\frac{1}{1-\rho}\!+\!\Lambda(\delta,\varepsilon)$.
Note that $C_{team} (\varepsilon,T) \!\to\! 0$ as $\varepsilon \!\to\! 0,\ T \!\to\! \infty$, and $C_{indp}(\varepsilon,T) \!\to\!\Lambda(\delta,\varepsilon)$ as $\varepsilon \!\to\! 0,\ T \!\to\! \infty$.
% For any $k >0$, by the dual variable update (\ref{dual_variable_update}) we can get
% \begin{align}
% \label{dual_variable_inclusion}
%     \vert \lambda_{k+1}^m-\lambda \vert&=\left \vert \left[ \lambda_k^m+\frac{1}{t_{\lambda},k}c^m \left( \underline{a}_k^m \right) \right]_+-\lambda \right\vert \nonumber\\ 
%     &\le \left\vert \lambda_k^m+\frac{1}{t_{\lambda},k}c^m-\lambda \right \vert \nonumber\\ 
%     & \le \vert \lambda_k^m-\lambda \vert+\frac{c_{max}}{t_{\lambda},k}
% \end{align}
% With (\ref{dual_variable_inclusion}), we have
% \begin{align}
%     \vert \lambda_{(n+1)}^m-\lambda_{(n)}^m \vert \le \frac{c_{max}T}{t_{\lambda,nT}} \le \frac{1}{n^2T}
% \end{align}
% Like (\ref{belief_accumulation}), we obtain
% \begin{align}
%     \lambda_{(n+1)}^m=\lambda_{(n)}^m+\beta_{(n)} e_{\lambda,(n)}^m ,
% \end{align}
% where $\vert e_{\lambda,(n)}^m \vert \le \frac{1}{\beta_{(n)}n^2T}$.
% According to (\ref{beta_properties}) and the fact that $\sum_{n=1}^{\infty} \frac{1}{n^2}<\infty$, we have $\lim_{n \to \infty}\frac{1}{\beta_{(n)}n^2}=0$.
% So there exists $N_{\epsilon}' \in \mathbb{N}_+$ such that
% \begin{align}
%     \vert e_{\lambda,(n)}^m \vert \le \frac{\epsilon}{T}=C_{\lambda}(\epsilon,T)
% \end{align}

% \subsubsection{Convergence Analysis}
% For the dual variables, we have
% \begin{align}
%     \vert \lambda_k^m \vert \le \vert \lambda_1^m \vert+\sum_{l=1}^{k-1}\frac{c_{max}}{l^2} \le \vert \lambda_1^m \vert+\sum_{l=1}^{\infty}\frac{c_{max}}{l^2}=\vert\lambda_1^m \vert+\frac{\pi^2}{6}\triangleq \lambda_{max}
% \end{align}
% \vspace{-0.5em}
We then analyze the convergence of the TNG.
To this end, we define the set-valued mapping
% \vspace{-3em}
\begin{align*}
    F(\pi) &\triangleq \bigg \{\!\! \Big(   br_{\tau}\left(\phi^m(\cdot,\pi^{-m})\right)\!-\!\pi^m\!+\!e^m \Big)_{m \in \mathcal{T}} \bigg|\  \forall m \in \mathcal{T}: \nonumber\\[-5pt]
    &\Vert e^m \Vert_1 \!\le\! C(\varepsilon,T), \ br_{\tau} \left( \phi^m(\cdot,\pi^{-m}) \right)\!+\!e^m \!\in\! \Delta\! \left( \underline{\mathcal{A}}^m \right) \!\!\bigg \}
\end{align*}
% \vspace{-2em}
for all $\pi \in \prod_{m \in \mathcal{T}} \Delta \left( \underline{\mathcal{A}}^m \right)  \triangleq \Pi $, where
% \vspace{-3em}
\begin{align*}
    C(\varepsilon,T)=
    \begin{cases}
        C_{\text{team}}(\varepsilon,T)+\varepsilon \quad \text{for DTOA},\\
        C_{\text{indp}}(\varepsilon,T)+\varepsilon \quad \text{for iDTOA}.
    \end{cases}
\end{align*}
% \vspace{-2em}
Then, (\ref{belief_accumulation}) and (\ref{11111}) imply that, for sufficiently large $n$,
% \vspace{-3em}
\begin{align*}
% \label{set_update}
    \pi_{(n+1)}\!-\!\pi_{(n)}\!-\!\beta_{(n)} \!\cdot \omega_{(n+1)} \in \beta_{(n)} \!\cdot F\left( \pi_{(n)} \right).
\end{align*}
% \vspace{-2em}
The following argument is similar to the proof in Appendix A.2 of \cite{donmez2024team}.
Finally, we get
% \vspace{-3em}
\begin{align*}
     \limsup_{k \to \infty}\!T\!N\!G(\pi_k^s)
     \!\!\le\!\!  \begin{cases}
        \!\tau \!\log\! \vert \mathcal{A} \vert, \!\!\!\!&\text{for DTOA},\\
        \!\tau \!\log\! \vert \mathcal{A} \vert\!+\!\vert \mathcal{T}  \vert^2 \overline{\phi}  \Lambda(\delta,\varepsilon_{\phi}), \!\!\!\!&\text{for iDTOA},
    \end{cases}
\end{align*}
% \vspace{-2em}
This completes the proof of Theorem \ref{convergence_of_gap}.

%% file: appendices_IEEE/proof_of_byzantine.tex
\subsection{Proof of Lemma \ref{theorem_of_checking_performance} and Theorem \ref{honest_TNG}}

\textbf{Proof of Lemma \ref{theorem_of_checking_performance}.}
Given a team $m \in \mathcal{T}$ and its supervisor $s$, define
$X_k^{m,s}\!\!=\!\!\mathbf{1}\{\!v_k^{m,s}\!\!=\!\!1,z_k^{m,s}\!\!=\!\mathrm{fail}\}$.
We have
% \vspace{-3em}
\begin{align*}
    \mathbb{P} \left( X_k^{m,s}\!=\!1 \right)\!=\!
    \begin{cases}
        q \!\cdot\! \eta_{FP}, \!\!\!\!& m \!\in\! \mathcal{H},\\
        q \!\cdot\! \left[(1\!-\!p_{\mathrm{lie}})\eta_{FP}\!+\!p_{\mathrm{lie}}(1\!-\!\eta_{FN}) \right]\!, \!\!\!\!&m \!\in\! \mathcal{B}.
    \end{cases}
\end{align*}
% \vspace{-2em}
For any $t\!>\!0$, $F_t^{m,s}\!\!=\sum_{k=1}^t X_k^{m,s}$ is the number of $\mathrm{fail}$ signals up to round $t$.
Using the KL-form Chernoff bound (see, e.g., \cite{wainwright2019high}, Chapter 2) and (\ref{frequency}), we get 
% \vspace{-3em}
\begin{align*}
    &\mathbb{P}\!\!\left( \!\!\frac{F_t^{m,s}}{t} \!\!>\!\!f \!\!\right) \!\!\le\!\! \exp \!\big(\!\!\! -\!\!tD(f \Vert q \eta_{FP}) \!\big),m \in \mathcal{H},\text{ and}\\
    &\mathbb{P}\!\!\left(\!\! \frac{F_t^{m,s}}{t}\!\! <\!\!f \!\!\right)\! \!\!\le\!\! \exp \!\Big(\!\!\! \!-\!\!tD \!\big(f \Vert q  [(\!1\!\!-\!\!p_{\mathrm{lie}}\!)\!\eta_{F\!P}\!\!+\!\!p_{\mathrm{lie}}\!(\!1\!\!-\!\!\eta_{F\!N}\!) ] \!\big) \!\!\Big)\!,m \!\!\in\!\! \mathcal{B},
\end{align*}
% \vspace{-2em}
where $D(x \Vert y) \triangleq x \log \frac{x}{y}+(1-x)\log \frac{1-x}{1-y}$.
This completes the proof of Lemma \ref{theorem_of_checking_performance}.

% \vspace{-0.5em}
\textbf{Proof of Theorem \ref{honest_TNG}.}
For $m \in \mathcal{H}$,
% \vspace{-3em}
\begin{align*}
    \mathbb{P}\left( \exists t\ge t_0: \frac{F_t^{m,s}}{t} >f \right) \\
    \le \sum_{t=t_0}^{\infty} \exp \big(-tD(f \Vert q \eta_{FP}) \big) \\
    =\frac{\exp \left( -t_0 D(f \Vert q \eta_{FP}) \right)}{1-\exp\left( -D(f \Vert q \eta_{FP}) \right)}.
\end{align*}
% \vspace{-2em}
We define the following events:
% \vspace{-3em}
\begin{align*}
    \mathcal{E}_0( t_0 )  =    \bigcap_{m \in \mathcal{H}}    \left\{  \forall t \ge  t_0  :  \frac{F_t^{m,s}}{t}   \le   f   \right\}  ,\\  \mathcal{E}_1( t_1 )  =    \bigcap_{m \in \mathcal{B}}   \left\{  \exists t \le  t_1   :  \frac{F_t^{m,s}}{t}    >    f  \right\}.
\end{align*}
% \vspace{-2em}
Then we obtain
% \vspace{-3em}
\begin{align*}
    &\mathbb{P} \left( \mathcal{E}_0(t_0) \right) \ge 1 - \vert \mathcal{H} \vert \cdot \frac{\exp \left( -t_0D(f \Vert q \eta_{FP}) \right)}{1-\exp \left( -D(f \Vert q \eta_{FP}) \right)} = 1 - \delta_{t_0},\\
    &\mathbb{P}   \left(  \mathcal{E}_1 ( t_1 )  \right)    \ge    1  -  \vert \mathcal{B} \vert   \cdot  \exp   \Big(     -  t_1 D  \big( f \Vert q   \big[ ( 1   -  p_{\mathrm{lie}} ) \eta_{F P} \\
    &\quad\quad\quad\quad\quad+  p_{\mathrm{lie}} ( 1   -  \eta_{F N} )  \big]  \big)   \Big)   =   1  -  \delta_{t_1}.
\end{align*}

% \vspace{-2em}
Define the good event as $\mathcal{E}( t_0, t_1 )  =  \mathcal{E}_{0}( t_0 )  \cap  \mathcal{E}_{1}( t_1 )$. Then we have
$$\mathbb{P} ( \mathcal{E}( t_0, t_1 ) )    \ge   1  - \delta_{t_0}   - \delta_{t_1}   =  1  - \delta_B.$$
On the good event $\mathcal{E}$, for $k\! >\!t_1$, the beliefs of Byzantine teams are frozen, while those of honest teams are not.
Following the proof of Theorem \ref{convergence_of_gap}, we can conclude
% \vspace{-3em}
\begin{align*}
     &\limsup_{k \to \infty}TNG_{\mathcal{H}}\left(\pi_k^s\right)\nonumber\\
     \le&
     \begin{cases}
        \tau \log \left\vert \mathcal{A}_{\mathcal{H}} \right\vert, &\text{for BR-DTOA},\\
        \tau \log \left\vert \mathcal{A}_{\mathcal{H}} \right\vert+\left\vert \mathcal{H}  \right\vert^2 \overline{\phi} \Lambda(\delta,\epsilon), &\text{for BR-iDTOA},
    \end{cases}
\end{align*}
% \vspace{-2em}
This completes the proof of Theorem \ref{honest_TNG}.